\documentclass[preprint]{imsart}

\RequirePackage[OT1]{fontenc}
\RequirePackage{amsthm,amsmath}
\RequirePackage[numbers]{natbib}
\RequirePackage[colorlinks,citecolor=blue,urlcolor=blue]{hyperref}

\startlocaldefs
\numberwithin{equation}{section}
\theoremstyle{plain}
\newtheorem{thm}{Theorem}[section]
\endlocaldefs

\usepackage{mystyle}
\usepackage{bigdelim}
\usepackage{rotating}
\usepackage{subfigure}
\graphicspath{{./images/}}

\begin{document}
\sloppy

\begin{frontmatter}
  \title{Estimation of Kullback-Leibler losses for noisy recovery problems
    within the exponential family}
  \runtitle{Estimation of Kullback-Leibler losses for the exponential family}

  \begin{aug}
    \author{\fnms{Charles-Alban} \snm{Deledalle}\thanksref{t1}%
      \ead[label=e1]{charles-alban.deledalle@math.u-bordeaux.fr}}
    \runauthor{C.-A. Deledalle}

    \address{IMB, CNRS, Universit\'e de Bordeaux, Bordeaux INP\\%
      351 cours de la Libération, F-33405 Talence, France\\
      \printead{e1}}

    \thankstext{t1}{The author would like to thank
J\'er\'emie Bigot and Jalal Fadili, as well as, the anonymous reviewers
for their generous help, relevant comments and constructive criticisms.
The author would also like to thank Janice Hau for proofreading this paper.}
  \end{aug}

  \begin{abstract}
    We address the question of estimating Kullback-Leibler losses
    rather than squared losses in recovery problems where
    the noise is distributed within the exponential family.
    Inspired by Stein unbiased risk estimator (SURE),
    we exhibit conditions under which these losses
    can be unbiasedly estimated or estimated with a controlled bias.
    Simulations on parameter selection problems in
    applications to image denoising and variable selection
    with Gamma and Poisson noises
    illustrate the interest of Kullback-Leibler losses and
    the proposed estimators.
  \end{abstract}

  \begin{keyword}[class=MSC]
    \kwd[Primary ]{62G05}
    \kwd{62F10}
    \kwd[; secondary ]{62J12}
  \end{keyword}

  \begin{keyword}
    \kwd{Stein unbiased risk estimator}
    \kwd{Model selection}
    \kwd{Kullback-Leibler divergence}
    \kwd{Exponential family}
  \end{keyword}

  \tableofcontents

\end{frontmatter}

\section{Introduction}

We consider the problem of predicting an unknown $d$-dimensional vector $\mu \in \RR^d$
from its noisy measurements $V \in \RR^d$. Given a collection
of parametric predictors of $\mu$, we focus on the selection of the predictor $\hat{\mu}$
that minimizes the discrepancy with the unknown vector $\mu$.
For instance, this includes the problem of selecting the best predictors from the
set of Least Absolute Shrinkage and Selection Operator (LASSO) solutions \cite{tibshirani1996regre} obtained for all possible choices of
regularization parameters.
To this end, the common approach is to select $\hat \mu$ that minimizes an unbiased estimate
of the expected squared loss $\EE \norm{\mu - \hat \mu}^2$,
typically, with the Stein unbiased risk estimator (SURE) \cite{stein1981estimation}.
Such estimators are classically built on some statistical modeling of the noise, e.g.,
as being distributed within the exponential family. In this context, we investigate
the interest of going beyond squared losses by
rather estimating a loss function grounded on an
information based criterion, namely, the Kullback-Leibler divergence.
We will first recall some basic properties of the
exponential family, give a quick review on risk estimation and
motivate the use of the Kullback-Leibler divergence.

\begin{table}[!t]
  \small
  \centering
  \caption{Examples of univariate distributions of the exponential family.
    The variable $y$ denotes an outcome of the random variable $Y$,
    $\mu = \EE[Y]$ is the unknown (location) parameter of interest,
    and the variables in brackets are known nuisance (scale and shape) parameters.}
  \hspace*{-0.85cm}%
  \begin{tabular}{@{}l@{\;\;\;\;}c@{\;\;\;\;}c@{\;\;\;}c@{\;\;\;}c@{}}
    \hline
    Distribution law & $\theta=\phi(\mu)$ & $\Lambda(\mu)$ & $h(y)$ & $A(\theta)$\\
    \hline
    \hline
    Gaussian ($\sigma > 0$)\\[-0.7em]
    $\quad \displaystyle \frac{1}{\sqrt{2\pi}\sigma} \exp\left(-\frac{(y-\mu)^2}{2\sigma^2}\right)$
    &
    $\dfrac{\mu}{\sigma^2}$
    &
    $\sigma^2$
    &
    $\displaystyle \frac{e^{-\frac{y^2}{2\sigma^2}}}{\sqrt{2\pi}\sigma}$
    &
    $\dfrac{\sigma^2 \theta^2}{2}$
    \\
    $\mu \in \RR$
    \\
    \hline
    Gamma ($L > 0$)\\
    $\quad \displaystyle \frac{L^L y^{L-1}}{\Gamma(L)\mu^L} e^{-\frac{Ly}{\mu}} \mathds{1}_{\RR^+}(y)$
    &
    $\dfrac{-L}{\mu}$
    &
    $\dfrac{\mu^2}{L}$
    &
    $\displaystyle \frac{L^L y^{L-1}}{\Gamma(L)}\mathds{1}_{\RR^+}(y)$
    &
    $-L \log ( -\theta/L )$
    \\
    $\mu > 0$
    \\
    \hline
    Poisson\\
    $\quad \displaystyle \frac{\mu^y e^{-\mu}}{y!}\mathds{1}_{\NN}(y)$
    &
    $\log \mu$
    &
    $\mu$
    &
    $\displaystyle \frac{\mathds{1}_{\NN}(y)}{y!}$
    &
    $\exp \theta$
    \\
    $\mu > 0$\\
    \hline
    Binomial ($n > 0$)\\
    $\quad \displaystyle \binom{n}{y} \mathtt{p}^y (1-\mathtt{p})^{n-y} \mathds{1}_{[0,n]}(y)$
    &
    $\log \dfrac{\mu}{n-\mu}$
    &
    $-\dfrac{\mu^2}{n} + \mu$
    &
    $\displaystyle \binom{n}{y}\mathds{1}_{[0,n]}(y)$
    &
    $n \log(1 + e^\theta)$
    \\
    $\mu = n\mathtt{p}$, $\mathtt{p} \in [0, 1]$
    \\
    \hline
    Negative Binomial ($r > 0$)\\
    $\quad \displaystyle \frac{\Gamma(r+y)}{y! \Gamma(r)} \mathtt{p}^y(1-\mathtt{p})^r \mathds{1}_{\NN}(y)$
    &
    $\log \dfrac{\mu}{r+\mu}$
    &
    $\dfrac{\mu^2}{n} - \mu$
    &
    $\dfrac{\Gamma(r+y)}{y! \Gamma(r)}\mathds{1}_{\NN}(y)$
    &
    $-r \log(1 - e^\theta)$
    \\
    \multicolumn{2}{l}{$\mu = r\mathtt{p}/(1-\mathtt{p})$, $\mathtt{p} \in [0, 1]$}
    \\
    \hline
    \hline
    \\
  \end{tabular}%
  \hspace*{-1cm}
  \label{tab:examples}
\end{table}

\paragraph{Exponential family.}
We assume that in the aforementioned recovery problem the noise distribution
belongs to the exponential family. Formally, the recovery problem can be reparametrized
using two one-to-one mappings $\psi : \RR^d \to \RR^d$
and $\phi : \RR^d \to \RR^d$ such that
$Y = \psi(V)$ has a probability measure $P_\theta$ characterized by a
probability density or mass function
with respect to the Lebesgue measure $\mathrm{d}y$
of the following form
\begin{align}\label{eq:klfamily}
  p(y; \theta) =
  h(y) \exp\left(\dotp{y}{\theta} - A(\theta) \right)
\end{align}
where $\theta = \phi(\mu) \in \RR^d$.
The distribution $P_\theta$ is said to be within the natural exponential family.
We call $\theta$ the natural parameter,
$Y$ a sufficient statistic for $\theta$,
$h : \RR^d \to \RR^+$ the base measure,
and $A : \RR^d \to \RR$ the log-partition function.
Classical and important properties of the exponential family include
$A$ that is convex,
$\EE[Y] = \nabla A (\theta)$ and $\Var[Y] = \nabla \nabla^t A (\theta)$
(see, e.g.,~\cite{brown1986fundamentals}).
Here and in the following, $\EE[Y] = \int Y \mathrm{d} P_\theta$ denotes the expectation
of the random vector $Y$ with respect to
the measure $\mathrm{d} P_\theta$,
and $\Var[Y] = \EE[(Y -  \EE[Y]) (Y -  \EE[Y])^t]$
is its so-called variance-covariance matrix.

Without loss of generality,
we consider that $Y$ is a minimal sufficient statistic. As a consequence,
$\nabla A$ is one-to-one and we can choose $\phi$ as
the canonical link function
satisfying $\phi = (\nabla A)^{-1}$ (as coined in the language of generalized linear models).
An immediate consequence is that $Y$ has expectation $\EE[Y] = \mu$ and
its variance is a function of $\mu$ given by $\Var[Y] = \Lambda(\mu)$ where
$\Lambda = (\nabla \nabla^t A) \circ \phi$.
The function $\Lambda : \RR^d \to \RR^{d \times d}$ is
the so-called variance function
(see, e.g.,~\cite{morris1982natural}),
also known as the noise level function
(in the language of signal processing).

Table \ref{tab:examples} gives five examples of univariate distributions of the
exponential family -- two of them are defined in a continuous domain,
the other three are defined in a discrete domain.

\paragraph{Risk estimation.}
We now assume that the predictor $\hat{\mu}$ of $\mu$ is a function of $Y$ only,
hence, we write it $\hat{\mu}(Y)$, and we focus on estimating
the loss associated to $\hat{\mu}(Y)$ with respect to $\mu$.
When the noise has a Gaussian distribution with independent entries,
SURE \cite{stein1981estimation} can be used
to estimate the mean squared error (MSE), or in short the risk, defined as:
$\mathrm{MSE}_\mu = \EE \norm{\mu - \hat{\mu}(Y)}^2$.
The resulting estimator, being independent on the unknown predictor $\mu$,
can serve in practice as an objective for parameter selection.
Eldar \cite{eldar2009robust} builds on Stein's lemma \cite{stein1981estimation},
a generalization of SURE valid for some continuous distributions of the exponential family.
It provides an unbiased estimate of the ``natural'' risk, defined as:
$\mathrm{MSE}_\theta = \EE \norm{\phi(\mu) - \phi(\hat{\mu}(Y))}^2$, i.e.,
the risk with respect to $\theta = \phi(\mu)$.
In the same vein, when the distribution is discrete, Hudson \cite{hudson1978nie}
provides another result for estimating the ``exp-natural'' risk:
$\mathrm{MSE}_\eta = \EE \norm{\exp \phi(\mu)  - \exp \phi(\hat{\mu}(Y))}^2$, i.e.,
the risk with respect to $\eta = \exp \theta$,
where $\exp : \RR^d \to \RR^d$ is the entry-wise exponential.
As $\phi$ is assumed one-to-one,
there is no doubt that if such loss functions cancel then
$\hat{\mu}(Y) = \mu$. In this sense, they provide
good objectives for selecting $\hat{\mu}(Y)$.
However, within a family of parametric predictors and without
strong assumptions on $\mu$, such a loss function might never cancel.
In such a case, it becomes unclear what its minimization leads it to select,
all the more when $\phi$ or $\exp \circ \, \phi$ are non-linear.
Furthermore, even when they are linear (e.g., $\exp \circ \,\phi = \text{id}$ for Poisson noise),
minimizing $\mathrm{MSE}_\mu = \EE \norm{\mu - \hat{\mu}(Y)}^2$
might not even be relevant
as it does not compensate for the heteroscedasticity of the noise
(this will be made clear in our experiments).
Estimating the reweighted or Mahanalobis risk given by $\EE \norm{\Lambda(\mu)^{-1/2} (\mu - \hat{\mu}(Y))}^2$
could be more relevant in this case, but its estimation is more intricate.

\paragraph{Kullback-Leibler divergence.}
The Kullback-Leibler (KL) divergence \cite{kullback1951information}
is a measure of information loss
when an alternative distribution $P_1$ is used to approximate the underlying one $P_0$.
Its formal definition is given by $  \Dd(P_0 \| P_1)
  =
  \int
  \ds P_0
  \log \frac{\ds P_0}{\ds P_1}
$.
Unlike squared losses,
it does not measure the discrepancy between an unknown parameter
and its estimate, but between the unknown distribution $P_0$ of $Y$
and its estimate $P_1$.
As a consequence, it is invariant with one-to-one reparametrization of the parameters
and, hence, becomes a serious competitor to squared losses.
Remark that it is also invariant under one-to-one transformations of $Y$ because
such transforms do not affect the quantity of information carried by $Y$.
Interestingly, provided $P_0$ and $P_1$ belongs to the same member of
the natural exponential family respectively with parameters $\theta_0$
and $\theta_1$,
the KL divergence can be written in terms
of the Bregman divergence associated with $A$ for points $\theta_0$
and ${\theta}_1$, i.e.,
\begin{align}\label{eq:kl}
  \Dd(P_0 \| P_1)
  =
  A({\theta}_1) - A(\theta_0) - \dotp{\nabla A(\theta_0)}{{\theta}_1 - \theta_0}.
\end{align}
While squared losses are defined irrespective of the noise distribution,
the KL divergence adjusts its penalty with respect to the scales and the shapes of the deviations.
In particular, it accounts for heteroscedasticity.

\paragraph{Contributions.}
In this paper, we address the problem of estimating KL losses,
i.e., losses based on the KL divergence.
As it is a non symmetric discrepancy measure,
we can define two KL loss functions.
The first one
\begin{equation}\tag{MKLA}
  \mathrm{MKLA} = \EE[\Dd(P_\theta \| P_{\hat{\theta}(Y)})]
\end{equation}
will be referred to as
the mean KL analysis loss as it can be given the following interpretation:
``how well might $P_{\hat{\theta}(Y)}$ explain independent copies of $Y$''.
The mean KL analysis loss is inherent to many statistical problems as it takes
as reference the true underlying distribution. It is at the heart of the maximum likelihood
estimator and is typically involved in non-parametric density estimation, oracle inequalities,
mini-max control, etc.~%
(see, e.g., \cite{hall1987kullback,george2006improved,rigollet2012kullback}).
The second one will be referred to as
the mean KL synthesis loss given by
\begin{equation}\tag{MKLS}
\mathrm{MKLS} = \EE[\Dd(P_{\hat{\theta}(Y)} \| P_\theta)]
\end{equation}
which can be given the following interpretation:
``how well might $P_{\hat{\theta}(Y)}$ generate independent copies of $Y$''.
The mean KL synthesis loss has also been considered in different statistical studies.
For instance, the authors of \cite{yanagimoto1994kullback} consider
this loss function to design a James Stein-like shrinkage predictor.
Hannig and Lee address a very similar problem to ours,
by designing a consistent estimator of $\mathrm{MKLS}$ used as an objective
for bandwidth selection in kernel smoothing problems subject to Gamma \cite{hannig2004kernel} and Poisson noise \cite{hannig2006poisson}.
Table \ref{tab:summary} gives a summary of our contributions.
It highlights which loss can be estimated and
under which conditions of the exponential family.
The main contributions of our paper are:
\begin{enumerate}
\item provided $y \mapsto \hat{\mu}(y)$ and the base measure $h$ are both weakly differentiable,
$\mathrm{MKLS}$ can be unbiasedly estimated
(Theorem \ref{thm:kls_continuous}),
\item for any mapping $y \mapsto \hat{\mu}(y)$,
$\mathrm{MKLA}$ can be unbiasedly estimated for Poisson variates
(Theorem \ref{thm:pukla}),
\item provided $y \mapsto \hat{\mu}(y)$ is $k \geq 3$ times differentiable
with bounded $k$-th derivative,
$\mathrm{MKLA}$ can be estimated with vanishing bias when $Y$ results from a large sample mean
of independent random vectors with finite $k$-th order moments
(Theorem \ref{thm:dkla}).
\end{enumerate}
It is worth mentioning that a symmetrized
version of the mean Kullback-Leibler loss: $\mathrm{MKLA} + \mathrm{MKLS}$,
can be estimated as soon as $\mathrm{MKLA}$ and $\mathrm{MKLS}$ can
both be estimated (e.g., for continuous distributions according to Table \ref{tab:summary}).

\begin{table}[!t]
\centering
\caption{Summary of what can be estimated provided $y \mapsto \hat{\mu}(y)$ is sufficiently regular.\hspace{1.5cm}{ }}
\label{tab:summary}
\begin{tabular}{lccc}
\cline{1-3}\\[-1em]
  & Continuous & Discrete\\
\cline{1-3}\\[-1em]
\cline{1-3}\\[-1em]
$\mathrm{MSE}_\mu$
& if $\phi(\mu) = \alpha \mu + \beta$ & if $\phi(\mu) = \log(\alpha \mu + \beta)$
&
\rdelim\}{4}{0.4cm}{
  \begin{sideways}\hspace{-1.2cm}\begin{minipage}{1.45cm}\centering Stein \&\\Hudson\end{minipage}\end{sideways}
}\\
& (Gaussian) & (Poisson)\\
\cline{1-3}\\[-1em]
$\mathrm{MSE}_{\theta}$ &
yes & \\
\cline{1-3}\\[-1em]
$\mathrm{MSE}_{\eta}$ & & yes\\
\cline{1-3}\\[-1em]
\cline{1-3}\\[-1em]
$\mathrm{MKLA}$
& if $\phi(\mu) = \alpha \mu + \beta$
& if $\phi(\mu) = \log(\alpha \mu + \beta)$
&
\rdelim\}{5}{0.4cm}{
  \begin{sideways}\hspace{-1.50cm}\begin{minipage}{1.65cm}\centering Our\\contributions\end{minipage}\end{sideways}
}\\
& (Gaussian) & (Poisson)\\
\cline{2-3}\\[-1em]
& \multicolumn{2}{c}{yes, when $Y$ results from a large sample mean}\\
& (Gaussian, Gamma, \ldots) & (Poisson, NegBin, Binomial, \ldots)\\
\cline{1-3}\\[-1em]
$\mathrm{MKLS}^*$ & yes & \\
\cline{1-3}\\[-1em]
\cline{1-3}\\[-1em]
\end{tabular}\\
$^*$consistently for kernel smoothing under Gamma \cite{hannig2004kernel} and Poisson noises \cite{hannig2006poisson}.\hspace{1.cm}{ }
\vspace{2em}
\end{table}

\section{Risk estimation under Gaussian noise}

This section recalls important properties of the $\mathrm{MSE}$
and the definition of SURE
under additive noise models of the form $Y = \mu + Z$
where $Z \sim \Nn(0, \sigma^2 \Id_d)$ and $\Id_d$ denotes the $d \times d$ identity matrix.

Before turning to the unbiased estimation of $\mathrm{MSE}_\mu$, it is important
to recall that for any additive models and zero-mean noise
with variance $\sigma^2 \Id_d$, provided the following quantities exists,
we have
\begin{align}\label{eq:variational_l2}
  \mathrm{MSE}_\mu
  =
  \underbrace{\EE \norm{Y - \hat{\mu}(Y)}^2 - d \sigma^2}_{\text{expected data fidelity}}
  +
  2 \underbrace{\tr \Cov(Y, \hat{\mu}(Y))}_{\text{model complexity}}
\end{align}
where $\Cov(Y, \hat{\mu}(Y)) = \EE[(Y -  \EE[Y]) (\hat{\mu}(Y) - \EE[\hat{\mu}(Y)])^t]$ is the cross-covariance matrix
between $Y$ and $\hat{\mu}(Y)$.
Equation \eqref{eq:variational_l2} gives
a variational interpretation of the minimization
of the $\mathrm{MSE}$ as the optimization of a trade-off between overfitting
(first term) and complexity (second term).
In fact, $\sigma^{-2} \tr \Cov(Y, \hat{\mu}(Y))$ is a classical
measure of the complexity of a statistical modeling procedure, known
as the degrees of freedom (DOF), see, e.g., \cite{efron1986biased}.
The DOF plays an important role in
model validation and model selection rules,
such as, Akaike information criteria (AIC) \cite{akaike1973information},
Bayesian information criteria (BIC) \cite{schwarz1978estimating},
and the generalized cross-validation (GCV) \cite{golub1979generalized}.

For linear predictors of the form $\hat{\mu}(y)=W y$, $W \in \RR^{d \times d}$
(think of least-square or ridge regression),
the DOF boils down to $\tr W$. As a consequence, the random quantity
$\norm{Y - \hat{\mu}(Y)}^2 - d \sigma^2 + 2 \tr W$
becomes an unbiased estimator of $\mathrm{MSE}_\mu$,
that depends solely on $Y$ without {\it prior} knowledge of $\mu$.
If $W$ is a projector, the DOF corresponds to the dimension of the target space,
and we retrieve the well known Mallows' $C_p$ statistic \cite{mallows1973some}
as well as the aforementioned AIC.
The SURE provides a generalization of these results that is
not only restricted to linear predictors but can be applied to weakly differentiable mappings.
A comprehensive account on weak differentiability can be found in e.g.,~\cite{EvansGariepy92,GilbargTrudinger98}.
Let us now recall Stein's lemma \cite{stein1981estimation}.

\begin{lem}[Stein lemma]\label{lem:stein}
Assume $f$ is weakly differentiable with essentially bounded weak partial derivatives
on $\RR^d$ and $Y \sim \Nn(\mu, \sigma^2 \Id_d)$, then
\begin{align*}
  \tr \Cov(Y, f(Y))
  &=
  \sigma^2 \EE\left[ \tr \pda{f(y)}{y}{Y} \right].
\end{align*}
\end{lem}

\noindent%
A direct consequence of Stein's Lemma,
provided $\hat{\mu}$ fulfills the assumptions of Lemma \ref{lem:stein}, is that
\begin{align}\label{eq:sure}
  \mathrm{SURE} =
  \norm{Y - \hat{\mu}(Y)}^2 - d \sigma^2 + 2 \sigma^2 \tr \pda{\hat{\mu}(y)}{y}{Y}
\end{align}
satisfies $\EE\mathrm{SURE} = \mathrm{MSE}_\mu$.
Applications of SURE emerged for choosing the smoothing
parameters in families of linear predictors \cite{li1985sure}
such as for model selection, ridge regression,
smoothing splines, etc.
After its introduction in the wavelet community with the
SURE-Shrink algorithm \cite{donoho1995adapting},
it has been widely used to various image restoration problems, e.g.,~%
with sparse regularizations
\cite{blu2007surelet,ramani2008montecarlosure,chaux2008nonlinear,pesquet-deconv,cai2009data,luisier2010sure,ramani2012regularization} or with non-local filters
\cite{vandeville2009sure,duval2011abv,dds2011nlmsap,vandeville2011non}.

\section{Risk estimation for the exponential family and beyond}

In this section, we recall how SURE has been extended
beyond Gaussian noises
towards noises distributed within the natural exponential family.

\paragraph{Continuous exponential family.}
We first consider continuous noise models, e.g., Gamma noise.
To begin, we recall a well known result
derived by Eldar \cite{eldar2009generalized},
that can be traced back to Hudson%
\footnote{In his paper, Hudson mentioned that Stein already knew about this result.}
in the case of independent entries \cite{hudson1978nie},
and that can be seen as a generalization of Stein's lemma.

\begin{lem}[Generalized Stein's lemma]\label{lem:hudson_stein}
Assume $f$ is weakly differentiable with essentially bounded weak partial derivatives on $\RR^d$
and $Y$ follows a distribution of the natural exponential family
with natural parameter $\theta$,
provided $h$ is also weakly differentiable on $\RR^d$, we have
\begin{align*}
  \EE \dotp{\theta}{f(Y)}
  &=
  - \EE\left[ \dotp{\frac{\nabla h(Y)}{h(Y)}}{f(Y)} + \tr \pda{f(y)}{y}{Y} \right].
\end{align*}
\end{lem}
\noindent
Lemma \ref{lem:hudson_stein}, whose proof can be found in \cite{eldar2009generalized},
provides an estimator of the dot product $\EE \dotp{\theta}{f(Y)}$
that solely depends on $Y$ without reference to $\theta$.
As a consequence, the Generalized SURE (as coined by \cite{eldar2009generalized}) defined by
\begin{align}\label{eq:gsure}
  \!\!\!
  \mathrm{GSURE}
  =
  \norm{\hat{\theta}(Y)}^2
  +
  2 \dotp{\frac{\nabla h(Y)}{h(Y)}}{\hat{\theta}(Y)}
  +
  2 \tr \pda{\hat{\theta}(y)}{y}{Y}
  \!\!\!
  +
  \frac{1}{h(Y)} \tr \pdda{h(y)}{y}{Y}
\end{align}
is an unbiased estimator of $\mathrm{MSE}_\theta$, i.e.,
$\EE\mathrm{GSURE} = \mathrm{MSE}_\theta$, provided
$\hat{\theta}$, $h$ and $\nabla h$ are weakly differentiable%
\footnote{Eq.~\eqref{eq:gsure} is obtained by applying Lemma \ref{lem:hudson_stein}
  on $\dotp{\theta}{\hat{\theta}(Y)}$, $\dotp{\theta}{\theta}$ and
  $\dotp{h(Y)^{-1} \nabla h(Y)}{\theta}$.}.
Note that omitting the last term in \eqref{eq:gsure}
leads to the seminal definition of GSURE given in \cite{eldar2009generalized}
which provides an unbiased estimate of $\mathrm{MSE}_\theta - \norm{\theta}^2$,
even though $\nabla h$ is not weakly differentiable.

The GSURE can be specified for Gaussian noise, and in this case
$\mathrm{GSURE} = \sigma^{-4} \mathrm{SURE}$
and the ``natural'' risk boils down to the risk as
$\mathrm{MSE}_\theta = \sigma^{-4} \mathrm{MSE}_\mu$.
In general, such a linear relationship between the ``natural'' risk and
the risk of interest might not be met. For instance, under Gamma noise%
\footnote{%
  A random variable $Y$ follows a Gamma distribution with scale parameter $L$
  if it results from the mean of $L$ independent and identically distributed
  exponential random variables.
  For this reason, $L$ is often referred to as the number of looks
  and controls the spread of the distribution as
  $\Var[Y] = \Lambda(\mu) = \frac{\mu^2}{L}$.
  This distribution is widely used to describe fluctuations of speckle
  in coherent laser imagery \cite{goodman1976some}.
}
with scale parameter $L$ (see Table \ref{tab:examples}),
with expectation $\mu$ and independent entries,
the GSURE reads as
\begin{align}\label{eq:gsure_gamma}
  \underset{\text{Gamma}}{\mathrm{GSURE}}
  =
  \sum_{i=1}^d
  \frac{L^2}{\hat{\mu}_i(Y)^2}
  -
  \frac{2 L (L-1)}{Y_i \hat{\mu}_i(Y)}
  +
  \frac{2 L}{\hat{\mu}_i(Y)^2} \pda{\hat{\mu}_i(y)}{y_i}{Y}
  +
  \frac{(L-1)(L-2)}{Y_i^2}
\end{align}
which, as soon as
$L>2$ and $\hat{\mu}$ fulfills the assumptions of Lemma \ref{lem:hudson_stein}, unbiasedly estimates
$\mathrm{MSE}_\theta = L^2 \EE\norm{\mu^{-1} - \hat{\mu}(Y)^{-1}}^{2}$,
where $(\cdot)^{-1}$ is the entry-wise inversion%
\footnote{$L\!>\!2$ implies that $h$ and $\nabla h$ are weakly differentiable.
  By omitting the last term of GSURE, an unbiased estimate of
  $L^2 \EE\norm{\mu^{-1} \!-\! \hat{\mu}(Y)^{-1}}^{2} \!-\! L^2 \norm{\mu^{-1}}^2$
  is obtained as soon as $L \!>\! 1$.}.
We will see in practice that minima of $\mathrm{MSE}_\theta$ can strongly
depart from those of interest.
As the GSURE can only measure discrepancy in the ``natural'' parameter space,
its applicability in real scenarios can thus be seriously limited.

\paragraph{Discrete exponential family.}
We now consider discrete noises distributed within the natural exponential family, e.g., Poisson
or binomial.
Before turning to the general result,
let us focus on Poisson noise with mean $\mu$
and independent entries for which
the Poisson unbiased risk estimator (PURE) defined as
\begin{align}
  \!\!
  \mathrm{PURE} \!=\! \norm{\hat{\mu}(Y)}^2 - 2 \dotp{Y}{\hat{\mu}_{\downarrow}(Y)} + \dotp{Y}{Y\!-\!1}
  \qwhereq \hat{\mu}_{\downarrow}(Y)_i \!=\! \hat{\mu}_i(Y \!-\! e_i),
\end{align}
unbiasedly estimates $\mathrm{MSE}_\mu$,
see, e.g., \cite{chen1975poisson,hudson1978nie}.
The vector $e_i$ is defined as $(e_i)_i=1$ and $(e_i)_j=0$ for $j\ne i$.
The PURE is in fact the consequence of the following lemma also due to Hudson \cite{hudson1978nie}.
\begin{lem}[Hudson's lemma]\label{lem:hudson}
Assume $Y$ follows a discrete distribution on $\ZZ^d$ of the
natural exponential family with natural parameter $\theta$, then
\begin{align*}
  \EE \dotp{\exp \theta}{f(Y)}
  =
  \EE \left[ \dotp{\frac{h_{\downarrow}(Y)}{h(Y)}}{f_\downarrow(Y)} \right]
  \qwhereq
  h_{\downarrow}(Y)_i = h(Y - e_i)
\end{align*}
holds for every mapping $f : \ZZ^d \to \RR$ where
$\exp$ is the entry-wise exponential.
\end{lem}
Hudson's lemma provides an estimator of the dot product $\EE \dotp{\exp \theta}{f(Y)}$
that solely depends on $Y$ without reference to the parameter $\eta = \exp \theta$.
As a consequence, we can define a Generalized PURE (GPURE) as
\begin{align}\label{eq:gpure}
  \!
  \mathrm{GPURE}
  =
  \norm{\exp \hat{\theta}(Y)}^2
  - 2 \dotp{\frac{h_{\downarrow}(Y)}{h(Y)}}{\exp \hat{\theta}_\downarrow(Y)}
  + \dotp{\frac{h_{\downarrow}(Y)}{h(Y)}}{\left(\frac{h_{\downarrow}}{h}\right)_{\!\!\downarrow}\!\!(Y)}
\end{align}
which unbiasedly estimates $\mathrm{MSE}_\eta$
for the discrete natural exponential family%
\footnote{
  Eq.~\eqref{eq:gpure} is obtained by applying three times Lemma \ref{lem:hudson}.
}.

As for GSURE, GPURE cannot in general measure discrepancy in the parameter space of interest,
and for this reason, its applicability in real scenarios can also be limited.
However, under Poisson noise, the ``exp-natural'' space coincides with the parameter space
of interest as $\eta = \exp(\phi(\mu)) = \mu$, hence, leading to the PURE.
Another interesting case, already investigated in \cite{hudson1978nie},
is the one of noise with a negative binomial distribution
with mean $\mu$ and independent entries, for which
the ``exp-natural'' space does not match with the one of $\mu$ but with the one of the underlying
probability vector $\mathtt{p} \in [0, 1]^d$ as defined in Table \ref{tab:examples} (we have $\theta_i = \log \mathtt{p}_i$).
In such a case, GPURE reads, for $r \in \RR^+_* / \{ 1, 2 \}$, as
\begin{align}
  \underset{\text{negbin}}{\mathrm{GPURE}} =
  \norm{\hat{\mathtt{p}}(Y)}^2
  -
  2
  \sum_{i=1}^d
  \displaystyle \frac{Y_i \hat{\mathtt{p}}_i(Y_i-1)}{Y_i + r - 1}
  +
  \sum_{i=1}^d
  \displaystyle \frac{Y_i (Y_i-1)}{(Y_i + r - 1)(Y_i + r - 2)}
\end{align}
and is an unbiased estimator of $\EE\left[\norm{\hat{\mathtt{p}}(Y) - \mathtt{p}}^2\right]$.

\paragraph{Other related works.}
It is worth mentioning that there have been several works focusing
on estimating mean squared errors in other scenarios.
For instance, when $Y$ has an elliptical-contoured distribution with a finite
known covariance matrix $\Sigma$, the works of
\cite{landsman2008stein,hamada2008capm}
provide a generalization of Stein's lemma that can also be used
to estimate the risk associated to $\mu$.
In \cite{raphan2007learning}, the authors provide a versatile approach that
provides unbiased risk estimators in many cases, including,
all members of the exponential family (continuous or discrete),
the Cauchy distribution,
the Laplace distribution,
and the uniform distribution
\cite{raphan2007learning}.
The authors of \cite{luisier2012cure} use a similar approach to
design such an estimator in the case of
the non-centered $\chi^2$ distribution
\cite{luisier2012cure}.

\section{Kullback-Leibler loss estimation for the exponential family}

We now turn to our first contribution
that provides, for continuous distributions of the natural exponential family,
an unbiased estimator of the Kullback-Leibler synthesis loss.

\begin{thm}[Stein Unbiased KLS estimator] \label{thm:kls_continuous}
Assume $y \mapsto \hat{\mu}(y)$ is weakly differentiable
with essentially bounded weak partial derivatives on $\RR^d$
and $Y$ follows a distribution of the natural exponential family with natural parameter $\theta$,
provided $h$ is weakly differentiable on $\RR^d$,
the following
\begin{align*}
  \mathrm{SUKLS}
  =
  \dotp{\hat{\theta}(Y) + \frac{\nabla h(Y)}{h(Y)}}{\hat{\mu}(Y)}
  +
  \tr \pda{\hat{\mu}(y)}{y}{Y}
  - A(\hat{\theta}(Y))
\end{align*}
where $\hat{\theta}(Y) = \phi(\hat{\mu}(Y))$,
is an unbiased estimator of $\mathrm{MKLS} - A(\theta)$.
\end{thm}

\begin{proof}
Remark that
$
\mathrm{MKLS}
=
\EE \left[ \dotp{\hat{\theta}(Y) - {\theta}}{\hat{\mu}(Y)} - A(\hat{\theta}(Y))\right]
+ A(\theta)
$ since $\nabla A(\hat{\theta}(Y)) = \hat{\mu}(Y)$.
Hence, Lemma \ref{lem:hudson_stein} leads to
\begin{align}
  \EE \left[ \dotp{\theta}{\hat{\mu}(Y)} \right]
  &=
  - \EE \left[ \dotp{\frac{\nabla h(Y)}{h(Y)}}{\hat{\mu}(Y)} +
    \tr \pda{\hat{\mu}(y)}{y}{Y} \right],
\end{align}
which concludes the proof.
\end{proof}

As GSURE, SUKLS can be specified for Gaussian noise,
and in this case $\mathrm{SUKLS} = (2\sigma^2)^{-1} (\mathrm{SURE} - \norm{Y}^2 + d\sigma^2)$
and the Kullback-Leibler synthesis loss boils down to the risk as
$\mathrm{MKLS} = (2\sigma^2)^{-1} \mathrm{MSE}_\mu$.
More interestingly, consider the following example of Gamma noise.
\begin{exmp}
Under Gamma noise with expectation $\mu$, shape parameter $L$ (as defined in Table \ref{tab:examples})
and independent entries,
SUKLS reads as
\begin{align}
  \underset{\text{Gamma}}{\mathrm{SUKLS}}
  =
  \sum_{i=1}^d
  \left[
  \frac{(L-1) \hat{\mu}(Y)_i}{Y_i}
  - L \log(\hat{\mu}(Y)_i)
  - L
  \right]
  +
  \tr \pda{\hat{\mu}(y)}{y}{Y}
\end{align}
which, up to a constant, and provided $L > 1$, unbiasedly estimates
\begin{align}
  \underset{\text{Gamma}}{\mathrm{MKLS}}
  =
  \sum_{i=1}^d \EE\left[L \frac{\hat{\mu}(Y)_i}{\mu_i}
    - L\log\left(\frac{\hat{\mu}(Y)_i}{\mu_i}\right)
    - L
  \right].
\end{align}
\end{exmp}
\noindent
In our experiments,
we will see that minimizing $\mathrm{MKLS}$ (or its SUKLS estimate)
leads to relevant selections, unlike
minimizing $\mathrm{MSE}_\theta$ (or its GSURE estimate).
Note that the authors of \cite{hannig2004kernel} have proposed
a consistent estimator of ${\mathrm{MKLS}}$ when $L=1$
(they did not study the case where $L>1$), their estimator
has been however designed only for kernel smoothing problems.\\

Theorem \ref{thm:kls_continuous} is a straightforward application of Lemma \ref{lem:hudson_stein}
that applies since $\mathrm{MKLS} - A(\theta)$ depends only on $\theta$
through a dot product $\dotp{\theta}{f(Y)}$ for some mappings $f$.
For discrete distributions, Lemma \ref{lem:hudson} only provides an
estimate of $\dotp{\exp(\theta)}{f(Y)}$ and hence cannot be applied
to estimate $\mathrm{MKLS}$.
Alternatively, we can focus on estimating the Kullback-Leibler analysis loss $\mathrm{MKLA}$.
To this end, a formula that provides an estimate of $\dotp{\nabla A(\theta)}{f(Y)}$ for some mappings $f$
is needed.
Of course, if $\nabla A(\theta) = \theta$ for some continuous distributions,
Lemma \ref{lem:hudson_stein} applies and can be used to design an estimator of $\mathrm{MKLA}$.
However, the only distribution
satisfying $\nabla A(\theta) = \theta$ is the normal distribution, for which
SURE can already be used to estimate $\mathrm{MKLA} = (2\sigma^2)^{-1} \mathrm{MSE}_\mu$.
More interestingly, if $\nabla A(\theta) = \exp(\theta)$ for some discrete distributions,
Lemma \ref{lem:hudson} applies and can be used to design an unbiased estimator of $\mathrm{MKLA}$.
The Poisson distribution satisfies this relation leading us to state the following
theorem.

\begin{thm}[Poisson Unbiased KLA estimator]\label{thm:pukla}
Assume $Y$ follows a Poisson distribution with expectation $\mu$
and independent entries, then
\begin{align*}
  \mathrm{PUKLA}
  =
  \norm{\hat{\mu}(Y)}_1 - \dotp{Y}{\log \hat{\mu}_\downarrow(Y)},
\end{align*}
is an unbiased estimator of $\underset{\text{Poisson}}{\mathrm{MKLA}} + \norm{\mu}_1 - \dotp{\mu}{\log \mu}$ where
\begin{align*}
\underset{\text{Poisson}}{\mathrm{MKLA}}
&= \EE \left[ \norm{\hat{\mu}(Y)}_1 - \dotp{\mu}{\log \hat{\mu}(Y) - \log \mu} \right] - \norm{\mu}_1
\end{align*}
and $\log$ is the entry-wise logarithm.
\end{thm}

\begin{proof}
The expression of $\mathrm{MKLA}$ follows directly from Table \ref{tab:examples}
and Equation \eqref{eq:kl}
since $\exp \theta = \mu$. From Lemma \ref{lem:hudson}, we get
\begin{align}
\EE\left[\dotp{\mu}{\log(\hat{\mu}(Y))}\right]
=
\EE\left[\dotp{\exp \theta}{\hat{\theta}(Y)}\right]
=
\EE\left[\dotp{\frac{h_{\downarrow}(Y)}{h(Y)}}{\hat{\theta}_\downarrow(Y)}\right],
\end{align}
which concludes the proof as $h_{\downarrow}(y)/h(y)=y$ and $\hat{\theta}_\downarrow(Y)=\log \hat{\mu}_\downarrow(Y)$.
\end{proof}

With such results at hand, only the Poisson distribution admits
an unbiased estimator of the mean Kullback-Leibler analysis loss.
In order to design an estimator of $\mathrm{MKLA}$ for a larger class of natural exponential distributions,
we will make use of the following proposition.

\begin{prop}\label{prop:kl_interp}
For any probability density or mass function $y\mapsto p(y ; \theta)$
of the natural exponential family of parameter $\theta$,
the Kullback-Leibler analysis loss
associated to $y \mapsto \hat{\theta}(y)$
can be decomposed as follows
\begin{align*}
  \mathrm{MKLA}
  =
  \underbrace{
    -\EE \log \frac{p(Y; \hat{\theta}(Y))}{p(Y ; \theta)}
  }_{\text{expected data fidelity loss}}
  +
  \underbrace{\tr \Cov\left( \hat{\theta}(Y), Y \right)}_{\text{model complexity}},
\end{align*}\\[-2em]
\begin{align*}
  \qwhereq
  -\EE \log \frac{p(Y; \hat{\theta}(Y))}{p(Y ; \theta)}
  &=
  \EE\left[ A(\hat{\theta}(Y)) - A(\theta) - \dotp{Y}{\hat{\theta}(Y) - \theta} \right]\\
  \qandq \quad
  \tr \Cov\left( \hat{\theta}(Y), Y \right)
  &=
  \EE \left[ \dotp{Y - \mu}{\hat{\theta}(Y)} \right]~.
\end{align*}
\end{prop}

\begin{proof}
Subtracting and adding $\dotp{Y}{\hat{\theta}(Y) - \theta}$ in the MKLA definition leads to
\begin{align*}
  \mathrm{MKLA}
  &=
  \EE\left[ A(\hat{\theta}(Y)) - A(\theta) - \dotp{Y}{\hat{\theta}(Y) - \theta} + \dotp{Y - \nabla A(\theta)}{\hat{\theta}(Y) - \theta} \right].
\end{align*}
As $-\log p(Y ; \theta) = -\log h(Y) - \dotp{Y}{\theta} + A(\theta)$
and $\nabla A(\theta) = \mu = \EE[Y]$, this concludes the proof.
\end{proof}

In the same vein as
for the decomposition \eqref{eq:variational_l2},
Proposition \ref{prop:kl_interp} provides a variational interpretation
of the minimization of $\mathrm{MKLA}$,
valid for noise distributions within the exponential family.
Minimizing $\mathrm{MKLA}$ leads to a maximum {\it a posteriori} selection
promoting faithful models with low complexity.
It boils down to \eqref{eq:variational_l2} when specified for Gaussian noise.
As for the $\mathrm{MSE}$, the fidelity term can always be unbiasedly estimated,
up to an additive constant, without knowledge of $\theta$.
Only the complexity term $\tr \Cov(\hat{\theta}(Y), Y)$, which
generalizes the notion of degrees of freedom,
is required to be estimated. Except for the Poisson distribution, none of the previous lemmas
can be applied to unbiasedly estimate this term.
However, we will show that it can be biasedly estimated,
with vanishing bias
depending on both the ``smoothness'' of $\hat{\theta}$
and the behavior of the moments of $Y$.
Towards this goal, let us first recall the Delta method.

\begin{lem}[Delta method]\label{lem:delta}
Let $Y_n = \tfrac{1}{n} (Z_1 + \ldots + Z_n)$, $n\geq 1$, where
$Z_1, Z_2, \ldots$ is an infinite sequence of independent and identically distributed
random vectors in $\RR^d$ with $\EE Z_i = \mu$, $\Var[Z_i] = \Sigma$
and finite moments up to order $k \geq 3$.
Let $f : \RR^d \to \RR$ be $k$ times totally differentiable
with bounded $k$-th derivative,
then
\begin{align*}
  \EE\left[f(Y_n) - f(\mu) \right]
  &=
  \frac{1}{2n}
  \tr\left( \Sigma \pdda{f(y)}{y}{\mu}\right)
  + O(n^{-2})
  =
  O(n^{-1}).
\end{align*}
\end{lem}
\noindent
Lemma \ref{lem:delta} is a direct $d$-dimensional extension of
\cite{lehmann1983theory} (Theorem~5.1a, page 109),
that allows us to introduce our biased estimator of $\mathrm{MKLA}$.

\begin{thm}[Delta KLA estimator]\label{thm:dkla}
Let $Y_n = \tfrac{1}{n} (Z_1 + \ldots + Z_n)$, $n \geq 1$, where
$Z_1, Z_2, \ldots$ is an infinite sequence of independent
random vectors in $\RR^d$ identically distributed within
the natural exponential family with natural parameter $\theta$,
log-partition function $A$, expectation $\mu$, variance function $\Lambda$
and finite moments up to order $k \geq 3$.
As a result, the distribution of $Y_n$ is also in
the natural exponential family parametrized
by $\theta_n = n \theta$ with log-partition function $A_n(\theta_n) = nA(\theta_n/n)$,
expectation $\mu$ and variance function $\Lambda_n = \Lambda / n$.
Provided $\hat{\theta}_n$ reads as $\hat{\theta}_n = n\hat{\theta}$, and
$\hat{\theta} : \RR^d \to \RR^d$ is $k$ times totally differentiable
with bounded $k$-th derivative,
then
\begin{align*}
  \mathrm{DKLA}_n &=
  A_n(\hat{\theta}_n(Y_n)) - \dotp{Y_n}{\hat{\theta}_n(Y_n)}
  +
  \tr\left(\Lambda_n(Y_n) \pda{\hat{\theta}_n(y)}{y}{Y_n}\right)\\
  \text{satisfies} \quad
  \EE \mathrm{DKLA}_n
  &=
  \mathrm{MKLA}_n - \dotp{\mu}{\theta_n} + A_n(\theta_n)
  + O(n^{-1})
\end{align*}
where $\mathrm{MKLA}_n$ is the KL analysis loss associated
to $\hat{\theta}_n$ with respect to $\theta_n$.
\end{thm}

\begin{proof}
Let $f(y) = \dotp{\hat{\theta}(y)}{y - \mu}$.
We have $f(\mu) = 0$ and $\pdda{f(y)}{y}{\mu} = 2 \pda{\hat{\theta}(y)}{y}{\mu}$.
Under the assumptions on $\hat{\theta}$,
the second-order approximation of Lemma \ref{lem:delta} applies
\begin{align}
  \!\tr \Cov(\hat{\theta}(Y_n), Y_n)
  \triangleq
  \EE\left[ f(Y_n) - f(\mu) \right]
  =
  \frac{1}{n}\tr\left( \Lambda(\mu) \pda{\hat{\theta}(y)}{y}{\mu}\right)
  + O(n^{-2})~.
\end{align}
Moreover, under the assumptions on $\hat{\theta}$
and as $\Lambda$ is in $\Cc^\infty$,
the first-order approximation of Lemma \ref{lem:delta} applies and
\begin{align}
  \EE\left[\tr\left( \Lambda(Y_n) \pda{\hat{\theta}(y)}{y}{Y_n} \right)\right]
  =
  \tr\left( \Lambda(\mu) \pda{\hat{\theta}(y)}{y}{\mu}\right)
  + O(n^{-1})~.
\end{align}
Subsequently, we have
\begin{align}
\EE
  &\mathrm{DKLA}_n -
\mathrm{MKLA}_n + \dotp{\nabla A_n(\theta_n)}{\theta_n} - A_n(\theta_n)
\nonumber
\\
&=
\EE\left[
  \tr\left( \Lambda_n(Y_n) \pda{\hat{\theta}_n(y)}{y}{Y_n}\right)
  - \tr \Cov\left( \hat{\theta}_n(Y_n), Y_n \right)
  \right]\\
&=
n
\EE\left[
  \frac{1}{n} \tr\left( \Lambda(Y_n) \pda{\hat{\theta}(y)}{y}{Y_n}\right)
  - \tr \Cov\left( \hat{\theta}(Y_n), Y_n \right)
  \right] = O(n^{-1})
\end{align}
and as $\nabla A_n(\theta_n) = \mu$,
this concludes the proof.
\end{proof}

It is worth mentioning that Theorem \ref{thm:dkla} can be applied to Gaussian noise,
with DKLA boiling down to SURE, as
$\mathrm{DKLA} =
(2 \sigma^2)^{-1}(\mathrm{SURE} - \norm{Y}^2 + d\sigma^2)
$.
However, the conclusion is not as strong,
as by virtue of Lemma \ref{lem:stein}, DKLA would be in fact an unbiased estimator
provided only that $\hat \mu$ is weakly differentiable.
More interestingly, consider the two following examples.

\begin{exmp}
Gamma random vectors $Y_n$ with expectation $\mu \in (\RR^+_*)^d$ and
shape parameter $L_n=n$ (as defined in Table \ref{tab:examples})
results from the sample mean of $n$ independent exponential random vectors
with expectation $\mu$
(entries of the vectors are supposed to be independent).
As exponential random vectors have finite moments,
provided $\hat{\mu}$ is sufficiently smooth and since $\phi$ is continuously
differentiable in $(\RR^+_*)^d$,
Theorem \ref{thm:dkla} applies and we get
\begin{align}
  \underset{\text{Gamma}\;\;}{\mathrm{DKLA}_n} &=
  \sum_{i=1}^d
  -L_n \log \hat{\mu}_i(Y_n)
  + \frac{L_n (Y_n)_i}{\hat{\mu}_i(Y_n)}
  +
  \frac{(Y_n)_i^2}{\hat{\mu}_i(Y_n)^2} \pda{\hat{\mu}_i(y)}{y_i}{Y_n}
  \nonumber\\
  \text{satisfies} \quad
  \EE \underset{\text{Gamma}\;\;}{\mathrm{DKLA}_n}
  &=
  \underset{\text{Gamma}\;\;}{\mathrm{MKLA}_n}
  + L_n \sum_{i=1}^d \log(\mu_i) - L_n
  + O(n^{-1})
  \\
  \qwhereq
  \underset{\text{Gamma}\;\;}{\mathrm{MKLA}_n}
  &=
  L_n
  \sum_{i=1}^d
  \EE \left[
  - \log(\hat{\mu}_i(Y_n)) + \frac{\mu_i}{\hat{\mu}_i(Y_n)} + \log(\mu_i) - 1
  \right].
  \nonumber
\end{align}
\end{exmp}

\begin{exmp}
Consider $Y_n$ the sample mean of $n$ independent
Poisson random vectors with expectation $\mu \in (\RR^+_*)^d$.
We have that $Y_n$, for all $n$, belongs to the natural exponential family
with $A_n(\theta_n) = n\exp(\theta_n/n)$ and $\theta_n = n \log \mu$
(entries of the vectors are supposed to be independent).
As Poisson random vectors have finite moments, provided $\hat{\mu}$ is sufficiently smooth
and since $\phi$ is continuously differentiable in $(\RR^+_*)^d$,
Theorem \ref{thm:dkla} applies and we get
\begin{align}
  \underset{\text{Poisson}\;\;}{\mathrm{DKLA}_n}
  &=
  n\norm{\hat{\mu}(Y_n)}_1 - \dotp{Y_n}{n \log \hat{\mu}(Y_n)
    + \diag\left( \pda{\log \hat{\mu}(y)}{y}{Y_n} \right)}
  \nonumber
  \\
  \text{satisfies} \quad
  \EE \underset{\text{Poisson}\;\;}{\mathrm{DKLA}_n}
  &= \underset{\text{Poisson}\;\;}{\mathrm{MKLA}_n}
  - n \dotp{\mu}{\log \mu} + n \norm{\mu}_1
  +
  O(n^{-1})
  \\
  \qwhereq
  \underset{\text{Poisson}\;\;}{\mathrm{MKLA}_n}
  &= n \EE \left[ \norm{\hat{\mu}(Y_n)}_1 - \dotp{\mu}{\log \hat{\mu}(Y_n) - \log \mu} - \norm{\mu}_1 \right].
  \nonumber
\end{align}
Interestingly, remark that
$\mathrm{PUKLA}(\hat{\mu}, Y) \approx \mathrm{DKLA}(\hat{\mu}, Y)$,
as soon as we have both $\hat{\mu}(Y-1) \approx \hat{\mu}(Y) - \hat{\mu}'(Y)$
and $|\hat{\mu}(Y)| \gg |\hat{\mu}'(Y)|$.
\end{exmp}

\section{Reliability study}

In this section, we aim at studying and comparing the sensitivity of the
previously studied risk estimators. Little is known about the variance of SURE:
$\Var[\mathrm{SURE}] = \EE\left[\left(\mathrm{SURE} - \mathrm{MSE}\right)^2\right]$.
It is in general an intricate problem and some studies \cite{pesquet-deconv,luisier-surelet} focus
instead on the reliability $\EE\left[\left(\mathrm{SURE} - \mathrm{SE} \right)^2\right]$
where $\mathrm{SE} = \norm{\mu - \hat{\mu}(Y)}^2$ (note that $\mathrm{MSE} = \EE[\mathrm{SE}]$).
Here, we do not aim at providing tight bounds on the reliability
as this would require specific extra assumptions for each pair of loss functions and estimators.
The next proposition provides only crude bounds on the reliability of each estimator.

\begin{prop}\label{prop:reliabilitis}
Assume $y \mapsto \hat{\theta}(y)$ is weakly differentiable.
Then, provided the following quantities are finite, we have
{\small%
\begin{align*}
  \frac{1}{2} \EE\left[\left(\overline{\mathrm{GSURE}}
  - \overline{\mathrm{SE}}_\theta
  \right)^2\right]^{1/2}
  &\leq
  \EE\left[
    \dotp{\frac{\nabla h(Y)}{h(Y)} + \theta}{\hat{\theta}(Y)}^2
    \right]^{1/2}
  \hspace{-0.5cm}&&+
  \EE\left[
    \left(
    \tr \pda{\hat{\theta}(y)}{y}{Y}
    \right)^2
  \right]^{1/2}\\
  \EE\left[\left(\overline{\mathrm{SUKLS}}
  - \overline{\mathrm{KLS}}
  \right)^2\right]^{1/2}
  &\leq
  \EE\left[
    \dotp{\frac{\nabla h(Y)}{h(Y)} + \theta}{\hat{\mu}(Y)}^2
  \right]^{1/2}
  \hspace{-0.5cm}&&+
  \EE\left[
    \left(
    \tr \pda{\hat{\mu}(y)}{y}{Y}
    \right)^2
  \right]^{1/2}\\
  \frac{1}{2}
  \EE\left[\left(\overline{\mathrm{PURE}}
  - \overline{\mathrm{SE}}_\mu
  \right)^2\right]^{1/2}
  &\leq
  \EE\left[
    \dotp{\mu}{\hat{\mu}(Y)}^2
  \right]^{1/2}
  \hspace{-0.5cm}&&+
  \EE\left[
    \dotp{Y}{\hat{\mu}_\downarrow(Y)}^2
    \right]^{1/2}\\
  \EE\left[\left(\overline{\mathrm{PUKLA}}
  - \overline{\mathrm{KLA}}
  \right)^2\right]^{1/2}
  &\leq
  \EE\left[
  \dotp{\mu}{\log \hat{\mu}(Y)}^2
  \right]^{1/2}
  \hspace{-0.5cm}&&+
  \EE\left[
    \dotp{Y}{\log \hat{\mu}_\downarrow(Y)}^2
  \right]^{1/2}\\
  \EE\left[\left(\overline{\mathrm{DKLA}}
  - \overline{\mathrm{KLA}}
  \right)^2\right]^{1/2}
  &\leq
  \EE\left[\dotp{Y - \mu}{\hat{\theta}(Y)}^2\right]^{1/2}
  \hspace{-0.5cm}&&+
  \EE\left[\tr\left(\Lambda(Y) \pda{\hat{\theta}(y)}{y}{Y}\right)^2\right]^{1/2}
\end{align*}}%
where $\mathrm{KLA} = \Dd(P_\theta \| P_{\hat{\theta}(Y)})$ (note that $\EE[\mathrm{KLA}] = \mathrm{MKLA}$),
and $\mathrm{KLS}$ is defined similarly.
The over-line refers to quantities for which additive constant
with respect to $\hat{\mu}(Y)$ are skipped, e.g.,
$\overline{\mathrm{SE}}_\mu = \mathrm{SE}_\mu - \norm{\mu}^2 = \norm{\hat{\mu}(Y)}^2 - 2 \dotp{\hat{\mu}(Y)}{\mu}$
and
$\overline{\mathrm{KLA}} = \mathrm{KLA} + A(\theta) - \dotp{\nabla A(\theta)}{\theta}$.
\end{prop}

\begin{proof}
This is a straightforward consequence of Cauchy-Schwartz's inequality.
\end{proof}

Proposition \ref{prop:reliabilitis} allows us to compare the relative
sensitivities of the different estimators. Comparing GSURE and SUKLS,
one can notice that the bounds are similar but the first one is controlled by $\hat{\theta}(Y)$
while the second one is controlled by $\hat{\mu}(Y)$.
While it is difficult to make a general
statement, we believe SUKLS estimates might be more stable than GSURE
since $\hat{\mu}(Y)$ has usually better control than $\hat{\theta}(Y)$,
given the non-linearity of the canonical link function $\phi$.

\section{Implementation details for the proposed estimators}

In this section, we explain how the proposed risk estimators
can be evaluated in practice within a reasonable computation time.\\

All risk estimators designed for continuous distributions
rely on the computation of
$\tr\left[ g(y) \left.\frac{\partial {f}(y)}{\partial y}\right|_{y} \right]$
for some mappings $g : \RR^d \to \RR^{d \times d}$ and $f : \RR^d \to \RR^d$.
For instance, SURE requires to compute such a a quantity
with $g(y) = \Id_d$ and $f = \hat{\mu}$ (see eq.~\eqref{eq:sure}).
In general, the computation of these terms requires at least $O(d^2)$ operations
and thus prevents the use of such risk estimators in practice.
Fortunately, following \cite{girard1989fast,ramani2008montecarlosure},
we can approximate such terms by using Monte-Carlo simulations,
thanks to the following relation
\begin{align}
  \tr\left[ g(y) \left.\frac{\partial {f}(y)}{\partial y}\right|_{y} \right]
  =
  \EE \dotp{\zeta}{g(y) \left.\frac{\partial {f}(y)}{\partial y}\right|_{y} \zeta}
  \quad \text{for} \quad
  \zeta \sim \Nn(0, \Id_d),
\end{align}
where the directional derivatives in the direction $\zeta \in \RR^d$
can be computed by using finite differences or algorithm differentiations
as described in \cite{deledalle2014stein}.
This leads in general to a much faster evaluation in $O(d)$ operations.\\

In the Poisson setting,
risk estimators rely on the computation of
$
\dotp{y}{f_\downarrow(y)}
$
for some mapping $f : \RR^d \to \RR^d$.
For instance, PUKLA requires to compute such a quantity
with $f = -\log \hat{\mu}$ (see Theorem \eqref{thm:pukla}).
Again, the computation of such terms requires at least $O(d^2)$ operations in general.
Based on first order expansions, we have empirically chosen to perform
Monte-Carlo simulations on the following approximation
\begin{align}
  \dotp{y}{f_\downarrow(y)}
  \approx
  \dotp{y}{
    f(y) - \diag\left( \left(\pda{f(y)}{y}{y} \zeta \right) \zeta^t\right)
  },
\end{align}
where $\zeta \in \{-1, +1\}^d$ is Bernoulli distributed with $\mathtt{p}=0.5$.
In our numerical experiments, this approximation led to $O(d)$ operations and
satisfactory results even though $f$ was chosen to be non-linear.
This approximation clearly deserves more attention but is considered
here to be beyond the scope of this study.

\section{Numerical experiments}
\label{sec:application}

In this section,
we will perform numerical experiments showing the interest
of the proposed Kullback-Leibler risk estimators in two different applications.

\subsection{Application to image denoising}

We first consider that $Y$ and $\mu$ are $d$ dimensional vectors representing images
on a discrete grid of $d$ pixels, such that entries with index $i$ are located at pixel
location $\delta_i \in \Delta \subset \ZZ^2$.
A realization $y$ of $Y$ represents a noisy observation of the image $\mu$.
The estimate $\hat{\mu}$ of $\mu$ is a denoised version of $y$.

\paragraph{Performance evaluation.}
In order to evaluate the proposed loss functions and their
estimates, visual inspection will be considered to assess the image quality
in terms of noise variance reduction and image content preservation.
In order to provide an objective measure of performance,
taking into account heteroscedasticity and tails of the noise, we will evaluate
the mean normalized absolute deviation error defined
as $\mathrm{MNAE} =  d^{-1} \sqrt{\pi/2} \norm{\Lambda(\mu)^{-1/2} (\mu - \hat{\mu}(Y))}_1$.
The MNAE measures to which extent $\hat{\mu}(Y)$ might belong to a confident interval
around $\mu$ with dispersion related to $\Lambda(\mu)$.
The MNAE is expected to be $1$ when $\hat{\mu}(Y) \sim \Nn(\mu, \Lambda(\mu))$,
and should get closer to $0$ when $\hat{\mu}(Y)$ improves on $Y$ itself.

\paragraph{Simulations in linear filtering.}
We consider here that $\hat{\mu}$ is the linear filter
\begin{align}
  \hat{\mu}(y) = W y \qwithq
  W_{i,j} = \frac{\exp(-\norm{\delta_i-\delta_j}^2/\tau^2)}{\sum_j \exp(-\norm{\delta_i-\delta_j}^2/\tau^2)},
\end{align}
where $W \in \RR^{d \times d}$
is a circulant matrix encoding a discrete convolution with a Gaussian kernel of bandwidth $\tau>0$.
In this context,
we will evaluate the relevance
of the different proposed loss functions and their estimates as objectives
to select a bandwidth $\tau$ offering a satisfying denoising.
\\

\begin{figure}[!t]
\centering%
\subfigure[Original image $\mu$]{\includegraphics[width=0.32\linewidth, viewport=125 339 380 432, clip]{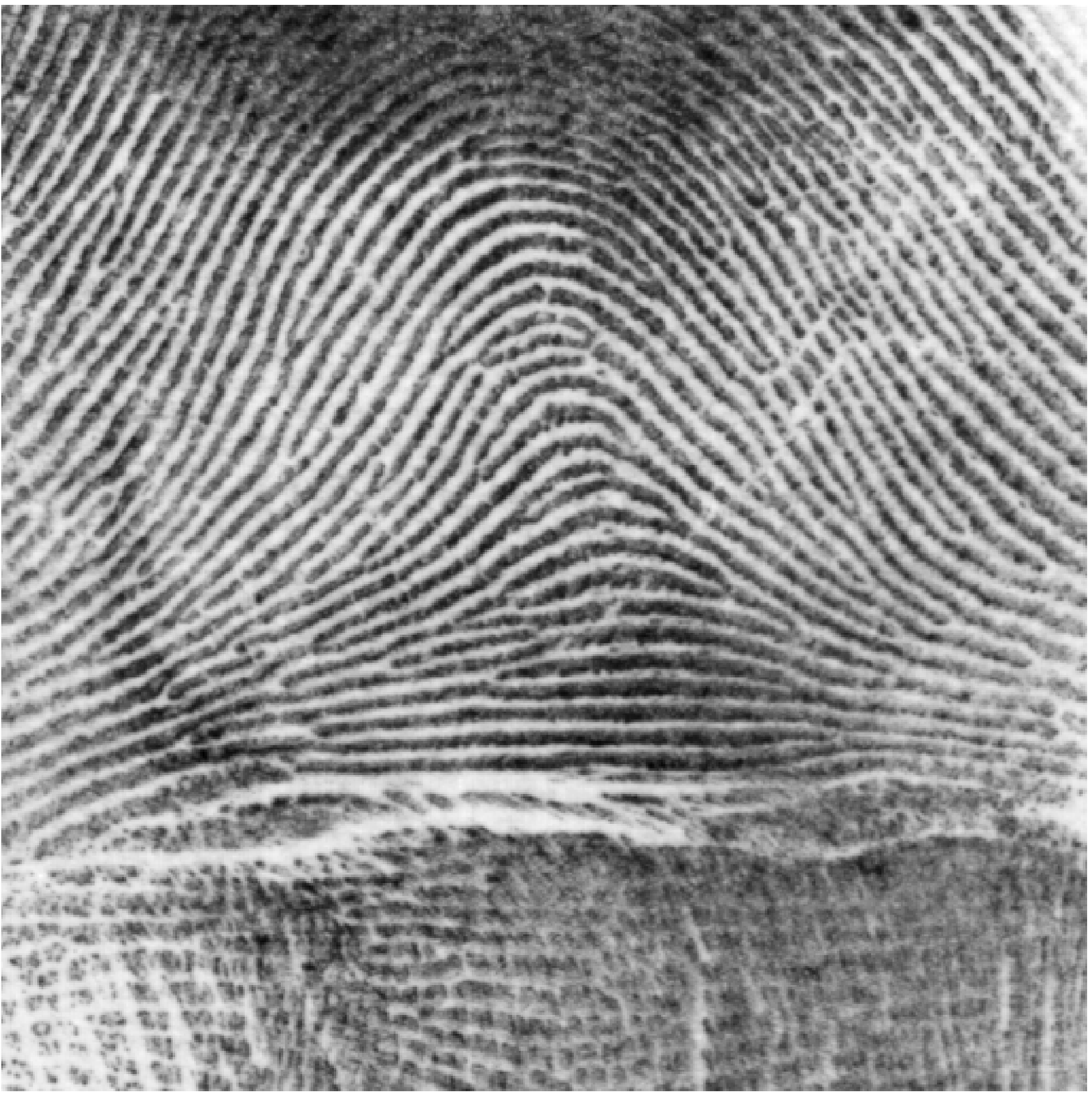}}\hfill%
\subfigure[Noisy image $y$]{\includegraphics[width=0.32\linewidth, viewport=125 339 380 432, clip]{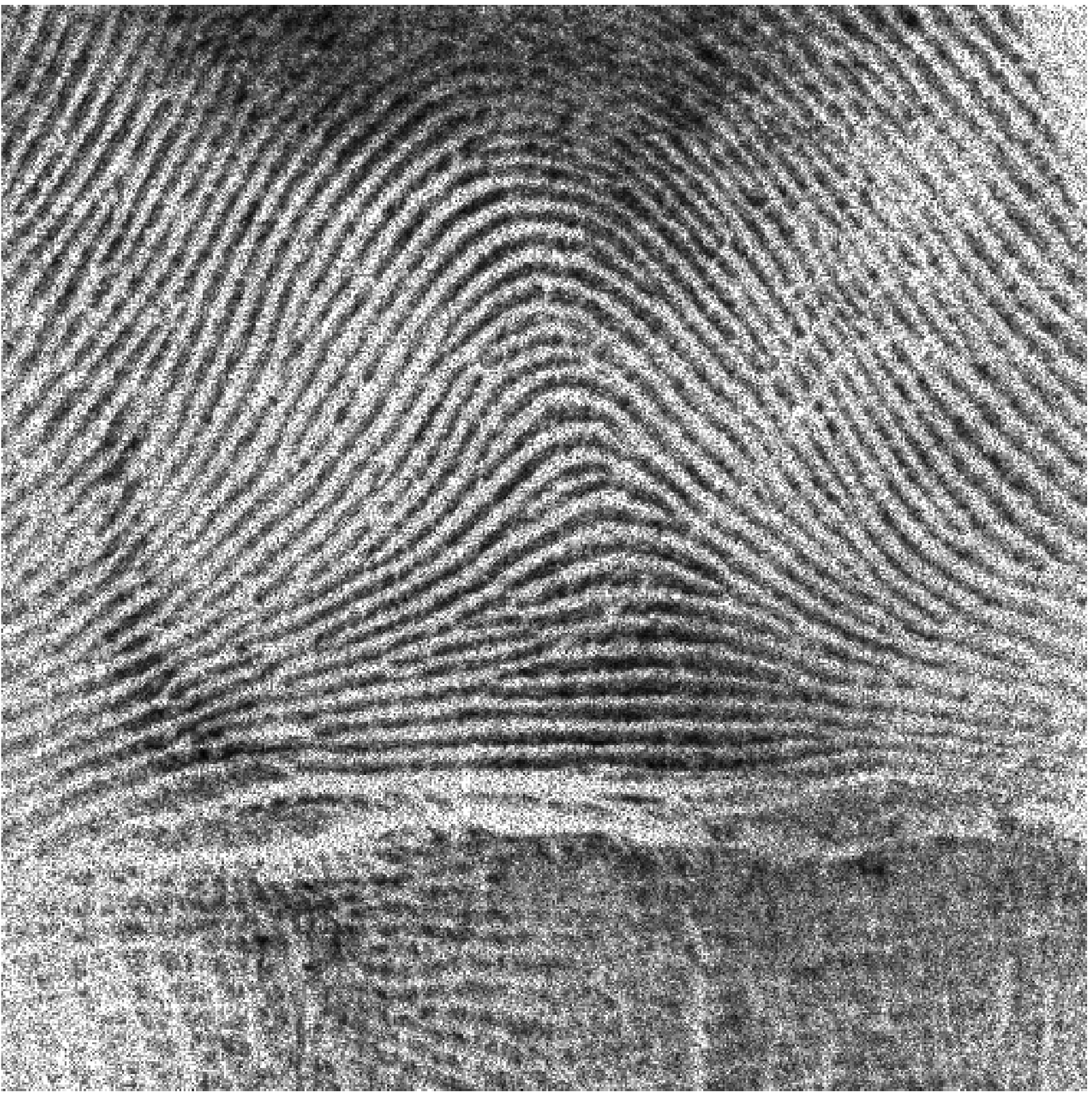}}\hfill%
\hspace{0.32\linewidth}{$ $}\\[-0.5em]
\subfigure[MNAE = 0.972]{\includegraphics[width=0.32\linewidth, viewport=125 339 380 432, clip]{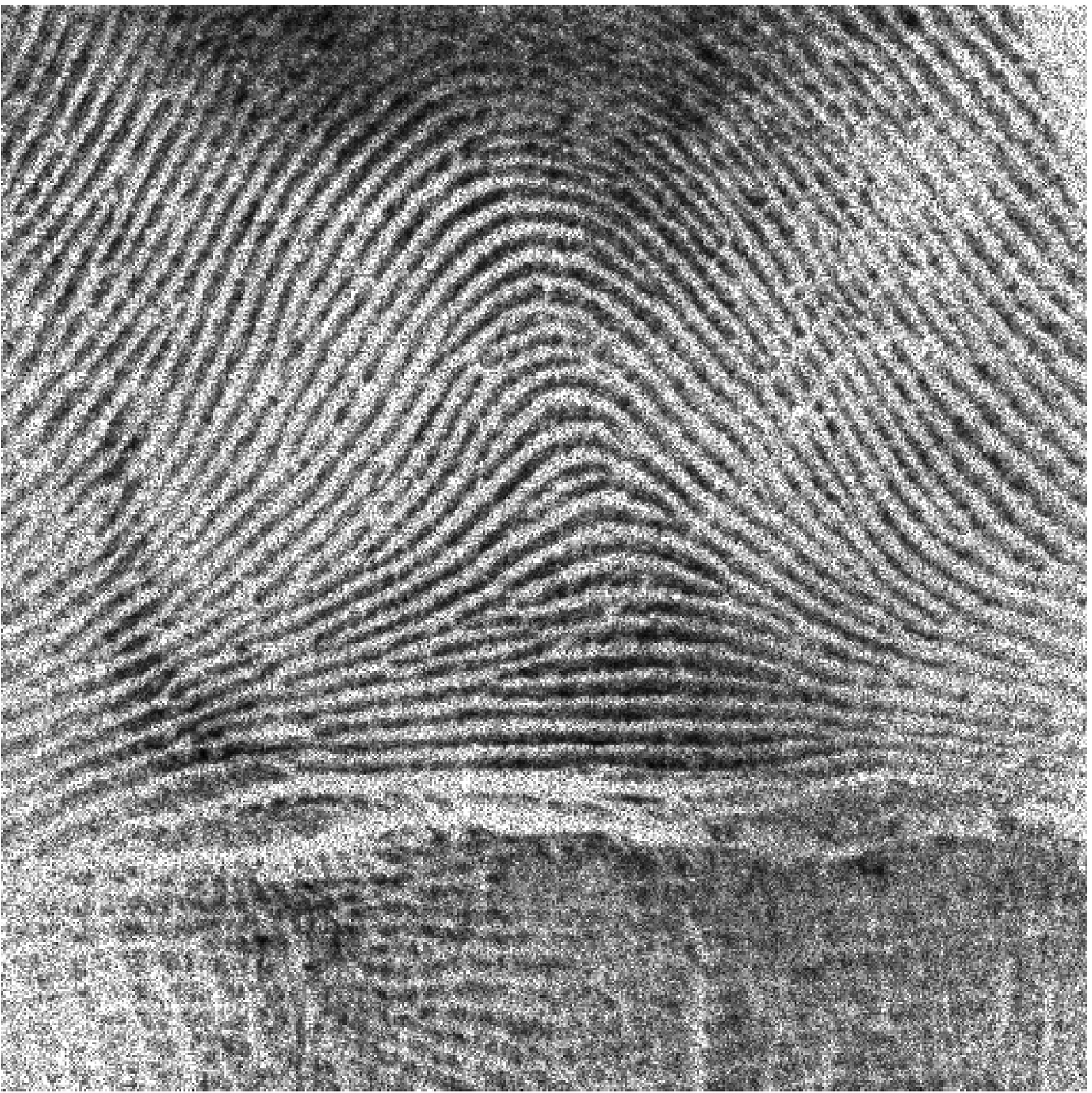}}\hfill%
\subfigure[MNAE = 0.496]{\includegraphics[width=0.32\linewidth, viewport=125 339 380 432, clip]{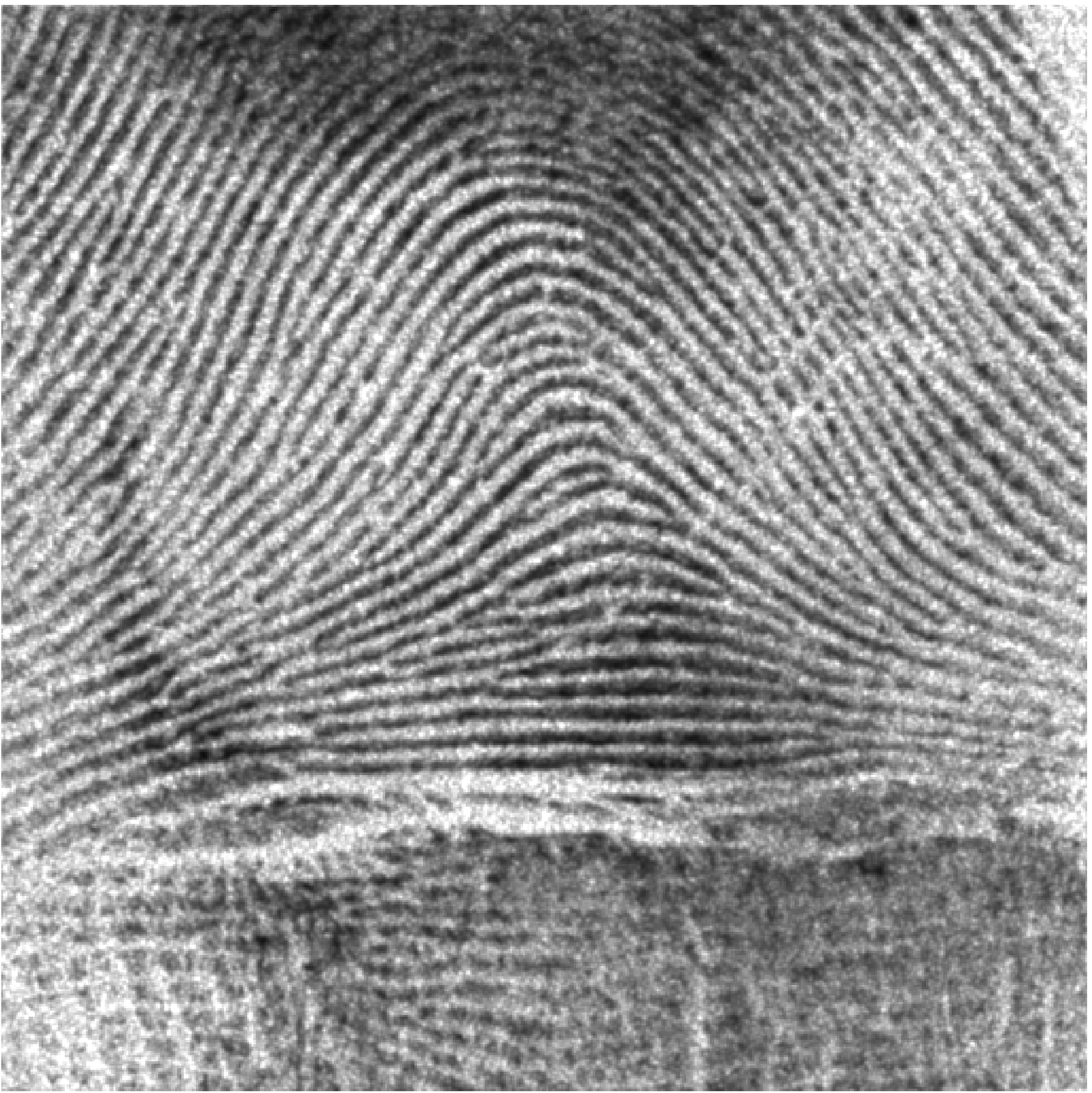}}\hfill%
\subfigure[MNAE = 0.500]{\includegraphics[width=0.32\linewidth, viewport=125 339 380 432, clip]{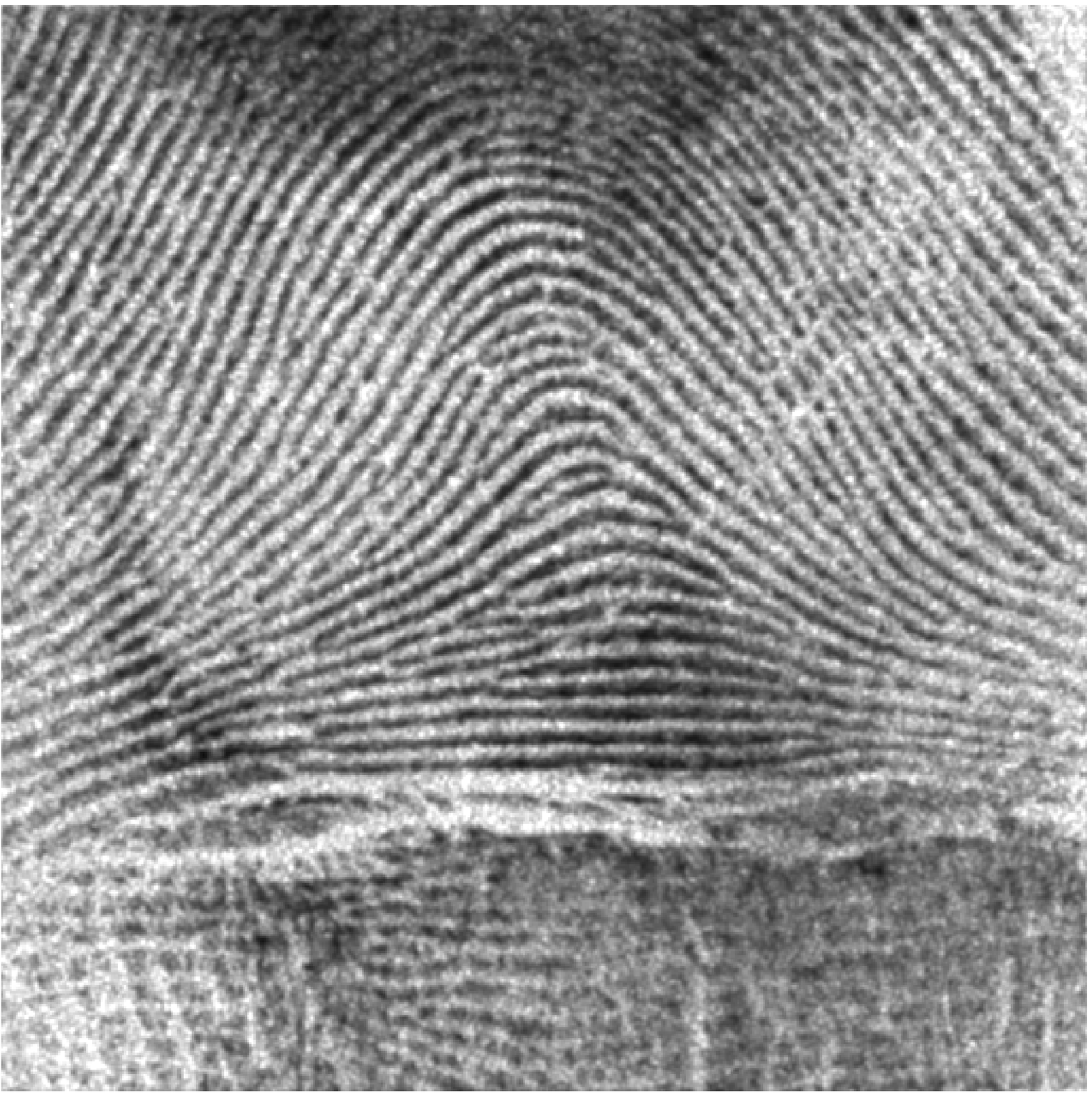}}\\[-0.5em]
\subfigure[$\tau^*=0.28$]{\includegraphics[width=0.32\linewidth]{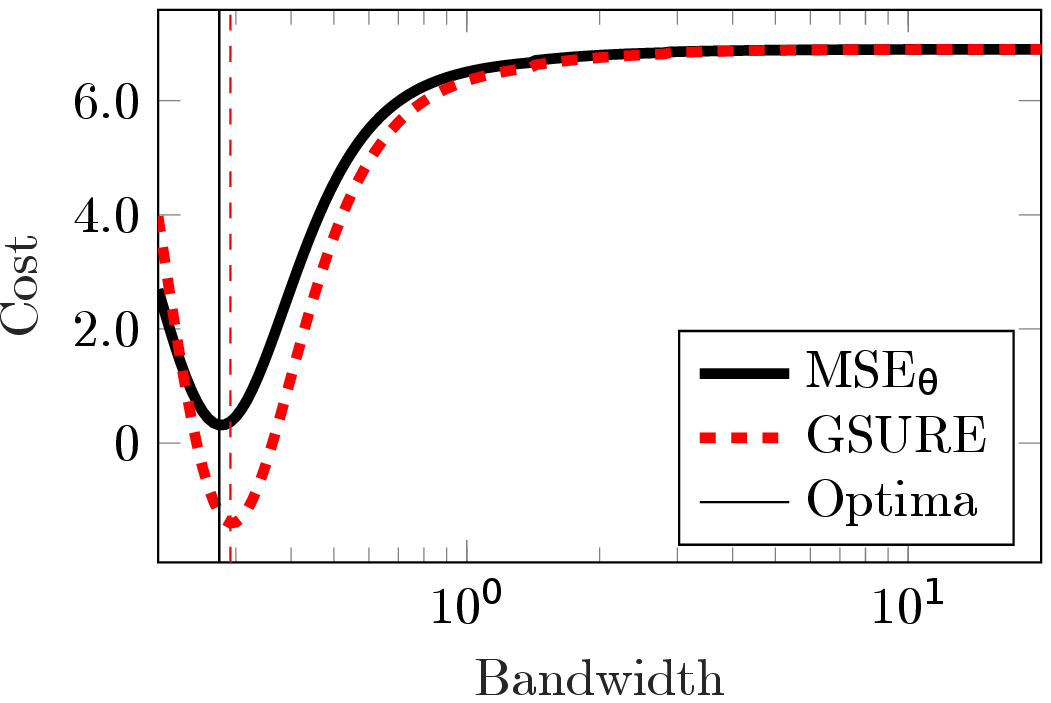}}%
\hfill%
\subfigure[$\tau^*=1.91$]{\includegraphics[width=0.32\linewidth]{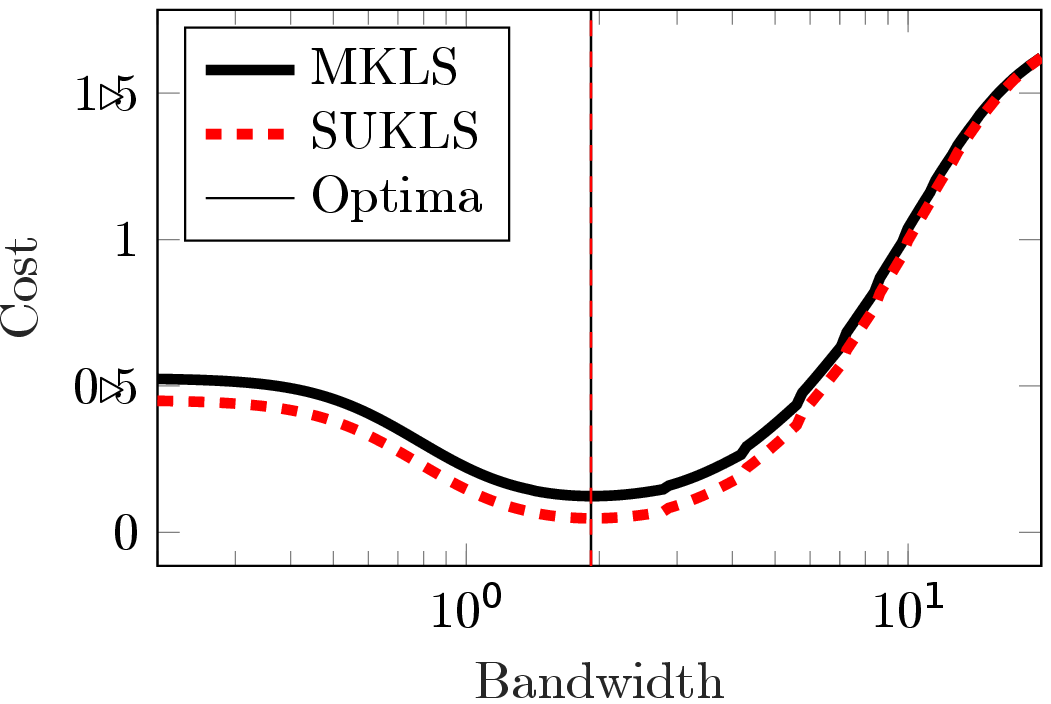}}%
\hfill%
\subfigure[$\tau^*=2.34$]{\includegraphics[width=0.32\linewidth]{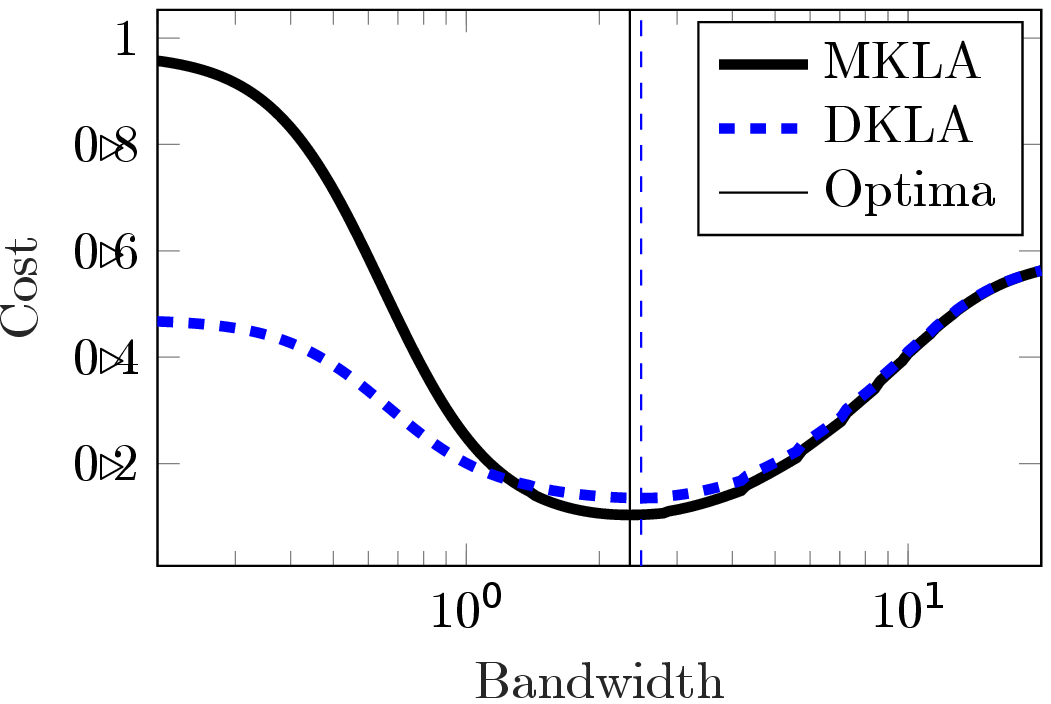}}\\
\caption{
  (a, b) Original and noisy image contaminated by Gamma noise with $L=3$
  (square-root of the images are displayed for better visual inspection).
  (c, d, e) Results of linear filtering
  for the optimal bandwidth with respect to the natural risk $\mathrm{MSE}_\theta$
  $\mathrm{MKLS}$ and $\mathrm{MKLA}$.
  The $\mathrm{MNAE}$ is indicated.
  (f,g,h)
  Risks and their estimates as a function of the bandwidth.
  The optimal bandwidth $\tau^\star$ is indicated.
  Red shows unbiased estimation and blue biased estimation.
}
\label{fig:linear_gamma}
\vspace{1em}
\end{figure}
\begin{figure}[!t]
\centering%
\subfigure[$\tau^*=0.20$/MNAE = 0.999]{\includegraphics[width=0.32\linewidth]{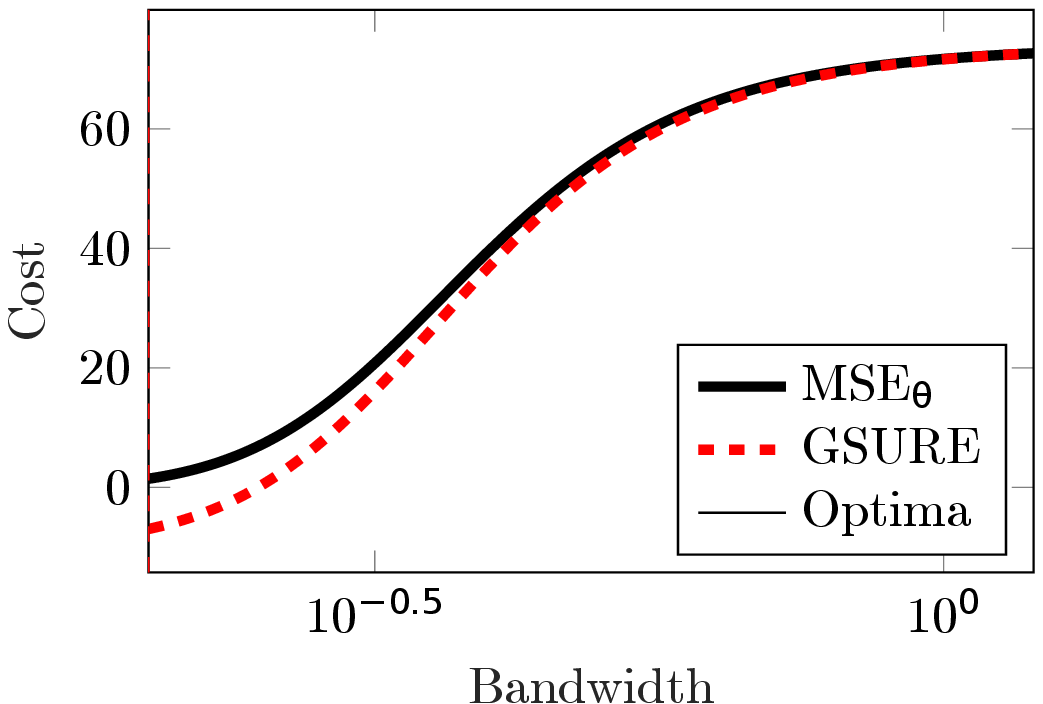}}%
\hfill%
\subfigure[$\tau^*=0.72$/MNAE = 0.904]{\includegraphics[width=0.32\linewidth]{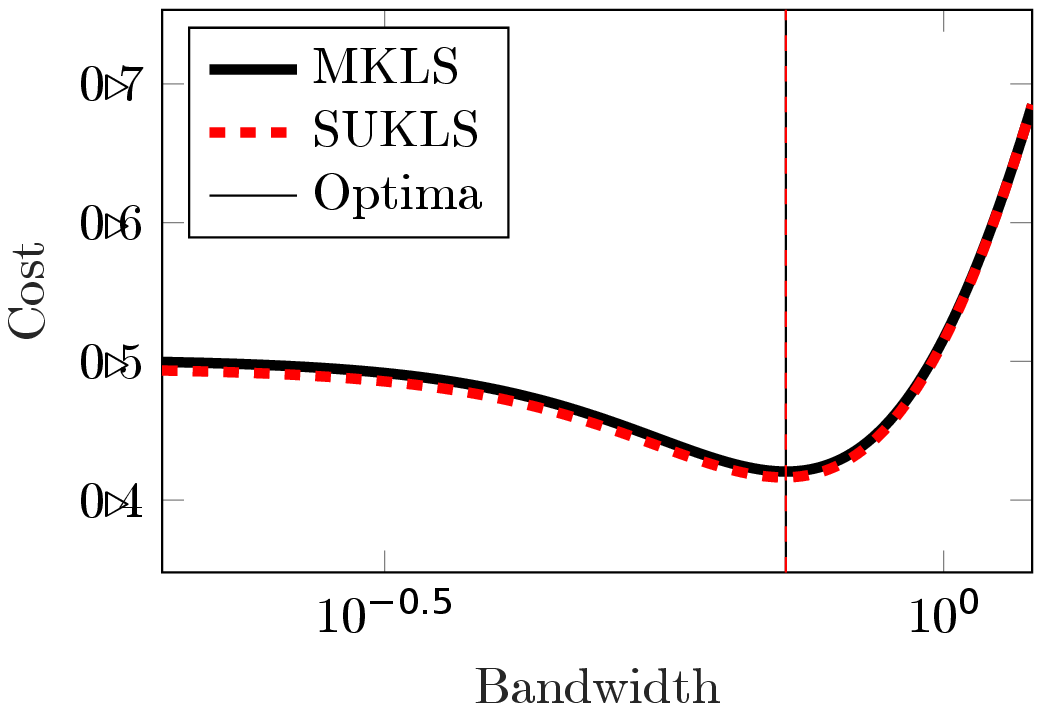}}%
\hfill%
\subfigure[$\tau^*=0.79$/MNAE = 0.900]{\includegraphics[width=0.32\linewidth]{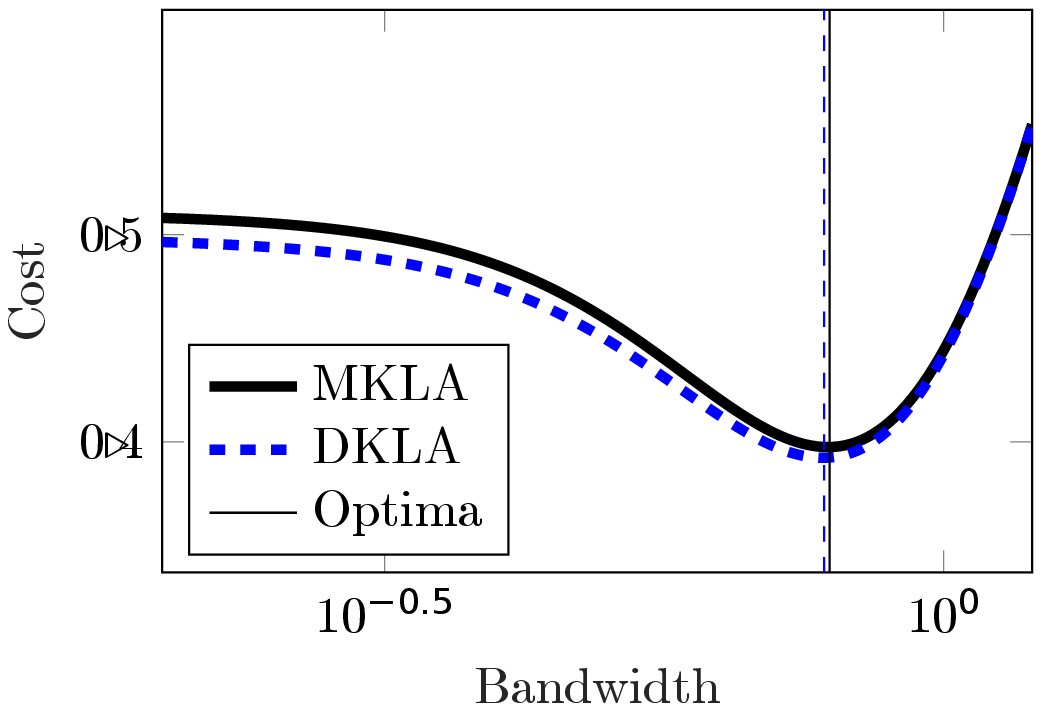}}\\
\caption{
  (a,b,c)
  Risks and their estimates as a function of the bandwidth in the same setting
  as in Figure \ref{fig:linear_gamma} but for Gamma noise with $L=100$.
  The optimal bandwidth $\tau^\star$ and the MNAE are indicated.
  Red shows unbiased estimation and blue biased estimation.
}
\label{fig:linear_gamma_L100}
\vspace{1em}
\end{figure}

Figure \ref{fig:linear_gamma} gives an example of a noisy observation $y$
of an image $\mu$ representing fingerprints whose pixel values are independently
corrupted by Gamma noise with shape parameter $L=3$.
We have evaluated the relevance of the natural risk $\mathrm{MSE}_\theta$
given by $\norm{\mu^{-1} - \hat{\mu}(Y)^{-1}}^2$,
$\mathrm{MKLS}$ and $\mathrm{MKLA}$ in selecting the bandwidth $\tau$.
Visual inspection of the results obtained at the
optimal bandwidth for each criterion shows that the natural risk $\mathrm{MSE}_\theta$ fails in selecting
a relevant bandwidth while $\mathrm{MKLS}$ and $\mathrm{MKLA}$ both provide a better trade-off.
The natural risk strongly penalizes small discrepancies at the lowest
intensities while not being sensitive enough for discrepancies at higher intensities.
As the noisy image has several isolated pixel values approaching $0$,
the natural risk will strongly penalize smoothing effects of such isolated structures
preventing satisfying noise variance reduction.
The Kullback-Leibler loss functions take into
account that Gamma noise has a constant signal to noise ratio. Hence, it does not favor
the restoration of either bright or dark structures more, allowing
satisfying smoothing for both, as assessed by the MNAE.
Finally, estimators of these loss functions with respectively GSURE, SUKLS and DKLA are given.
Note that for $L=3$, the Gamma distribution is far from reaching the asymptotic
conditions of Theorem \ref{thm:dkla}. As a result, bias is
not negligible (it becomes obvious for the lowest values of $\tau$ in Figure \ref{fig:linear_gamma}.h).
Nevertheless, minimizing DKLA can still provide an accurate location of the optimal parameter
for $\mathrm{MKLA}$.

Figure \ref{fig:linear_gamma_L100} reproduces the same experiment but with
Gamma noise with $L=100$, i.e., with a much better signal to noise ratio.
Interestingly, the bias of DKLA
gets much smaller than in the previous experiment.
This was indeed expected as with $L=100$,
the Gamma distribution fulfills much better the asymptotic conditions
of Theorem \ref{thm:dkla}. Remark that MNAE values are still in favor
of Kullback-Leibler objectives, but the gains are much smaller.
In fact, all MNAE values get closer to $1$ since noise reduction
with signal preservation using linear filtering
becomes much harder in such a low signal to noise ratio setting.

\begin{figure}[!t]
\centering
\subfigure[Original image $\mu$]{\includegraphics[width=0.32\linewidth,viewport=100 70 356 163,clip]{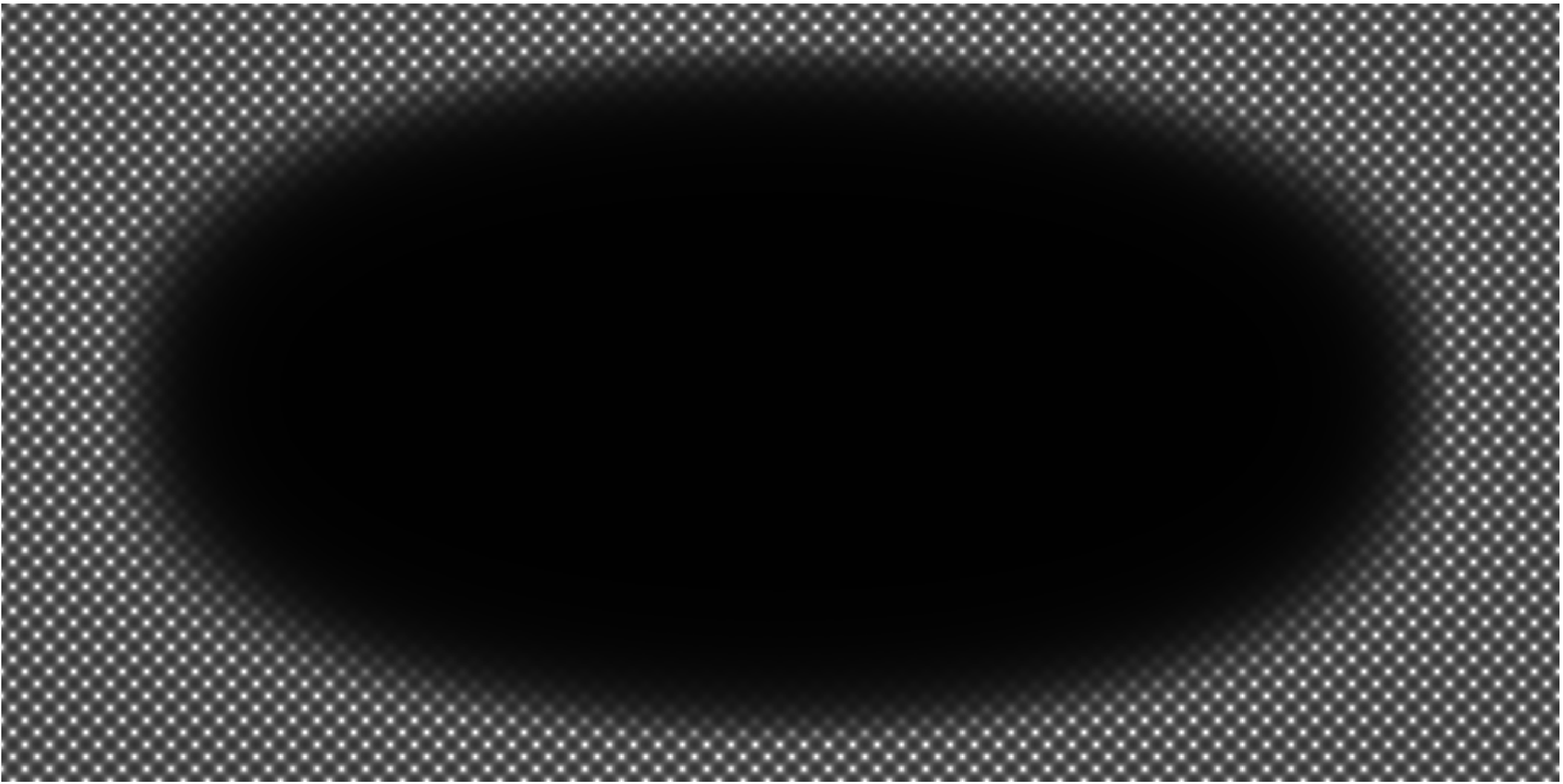}}\hfill%
\subfigure[Noisy image $y$]{\includegraphics[width=0.32\linewidth,viewport=100 70 356 163,clip]{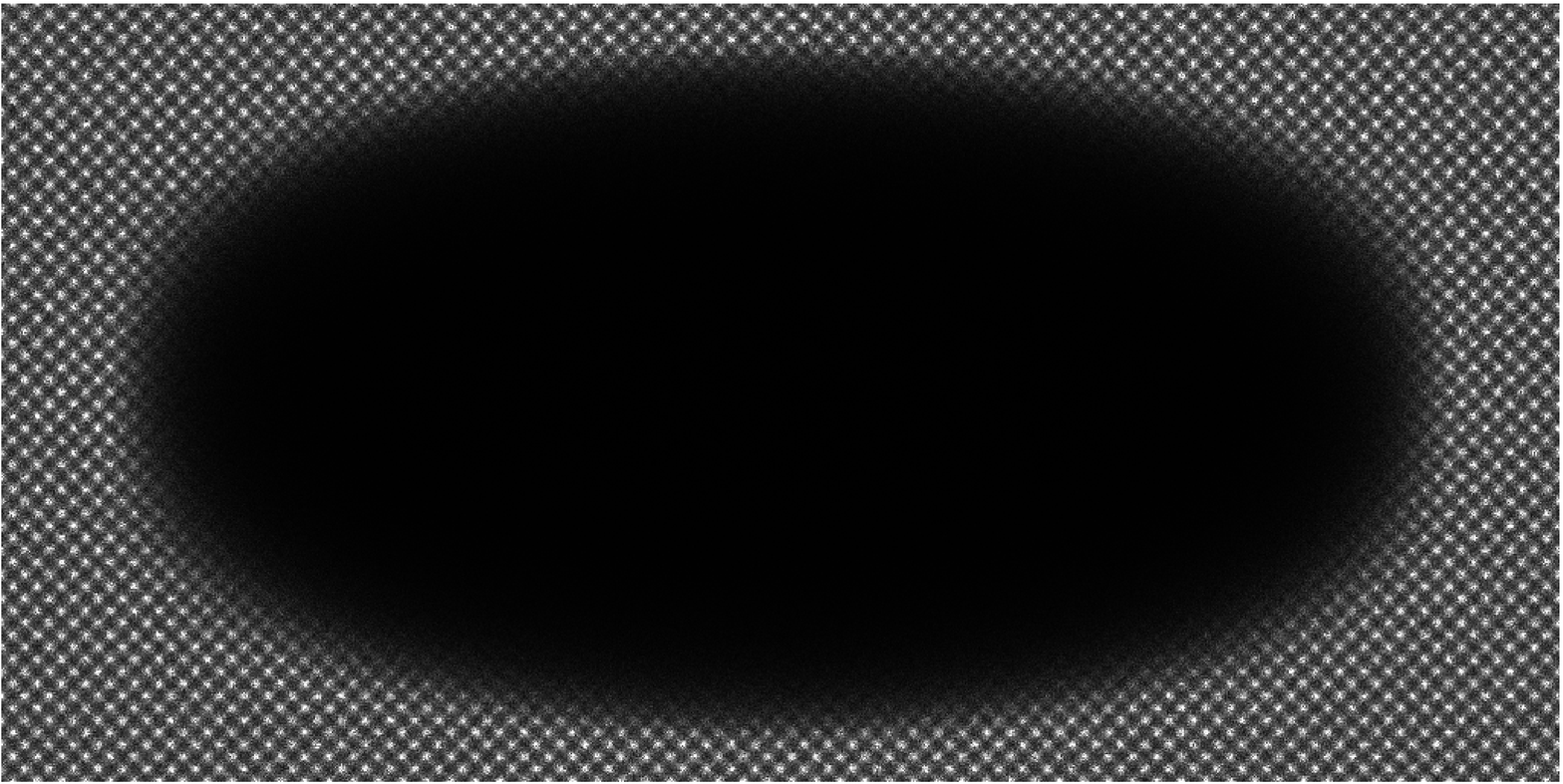}}\hfill%
\hspace{0.32\linewidth}{$ $}\\[-0.5em]
\subfigure[MNAE = 0.853]{\includegraphics[width=0.32\linewidth,viewport=100 70 356 163,clip]{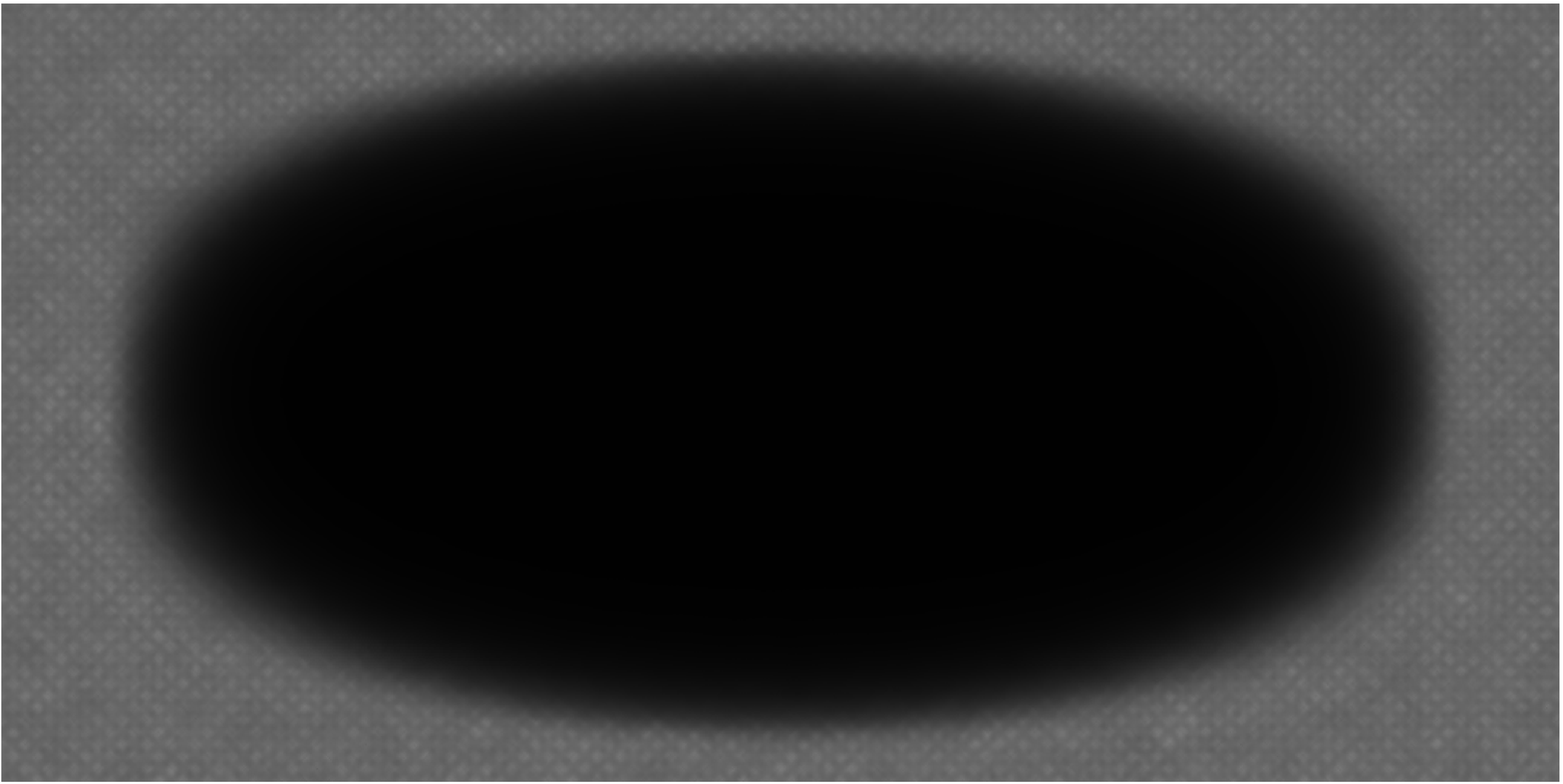}}\hfill%
\subfigure[MNAE = 0.249]{\includegraphics[width=0.32\linewidth,viewport=100 70 356 163,clip]{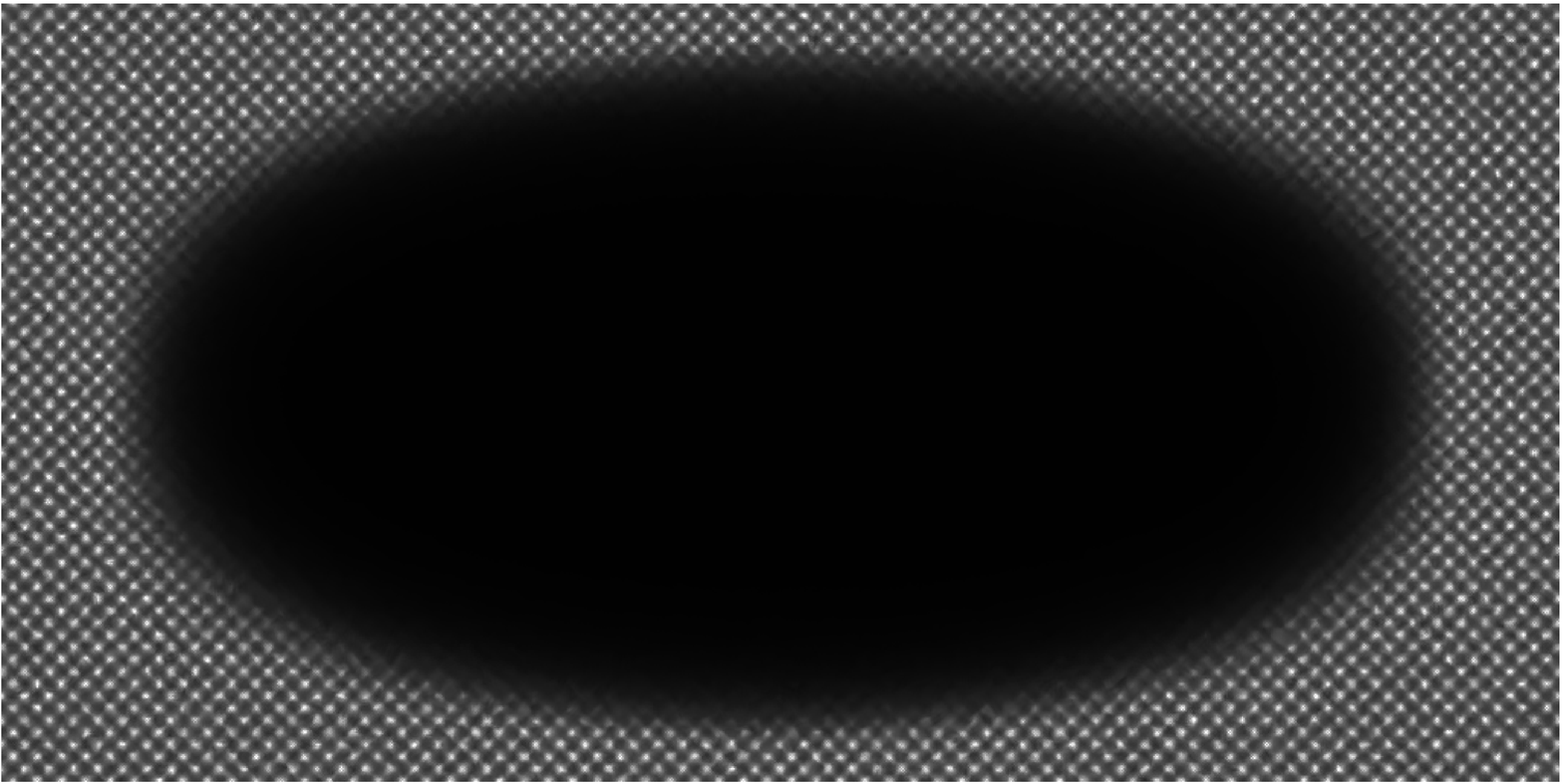}}\hfill%
\subfigure[MNAE = 0.249]{\includegraphics[width=0.32\linewidth,viewport=100 70 356 163,clip]{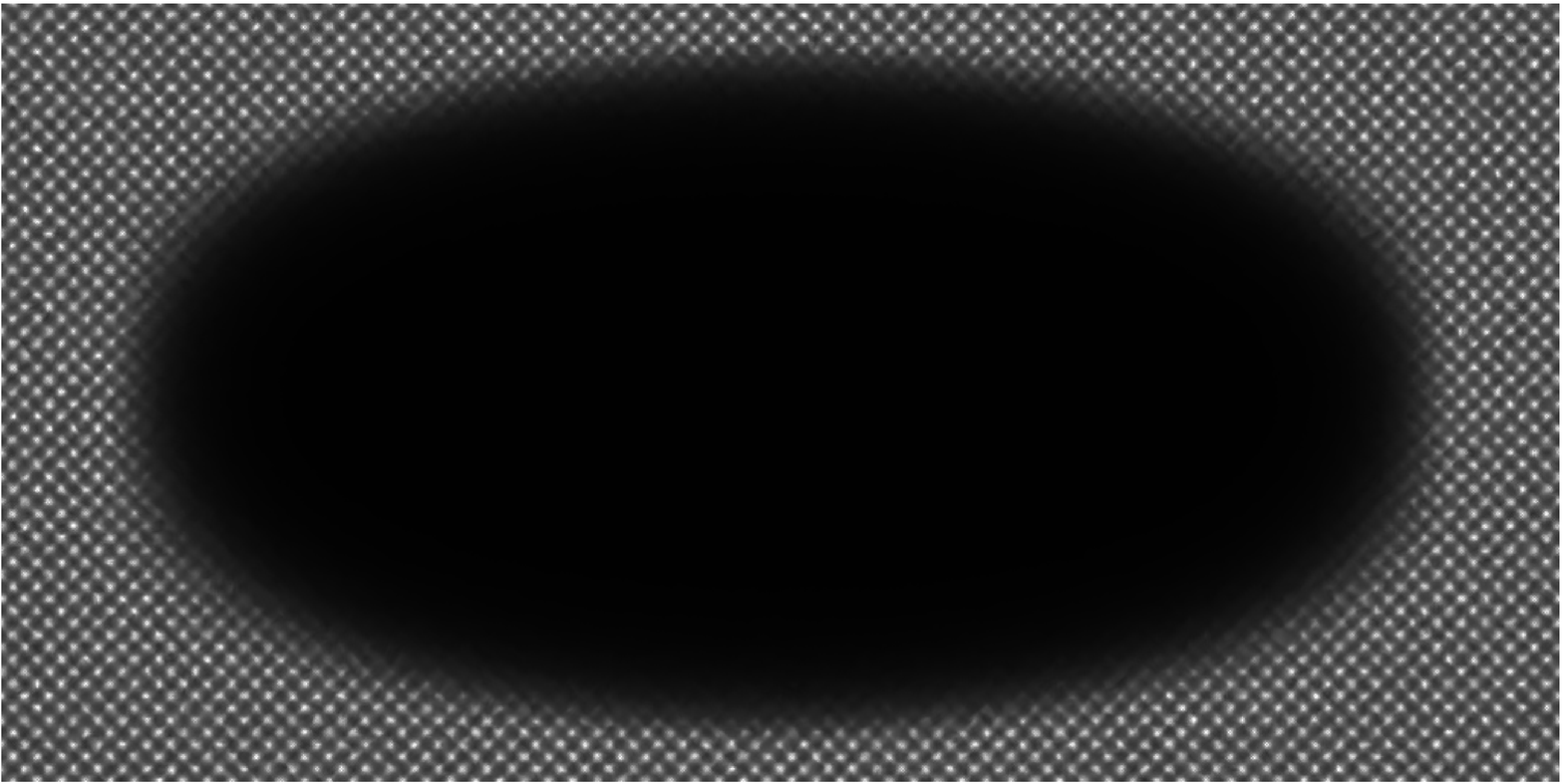}}\\[-0.5em]
\subfigure[$\tau^\star=23.84$]{\includegraphics[width=0.32\linewidth]{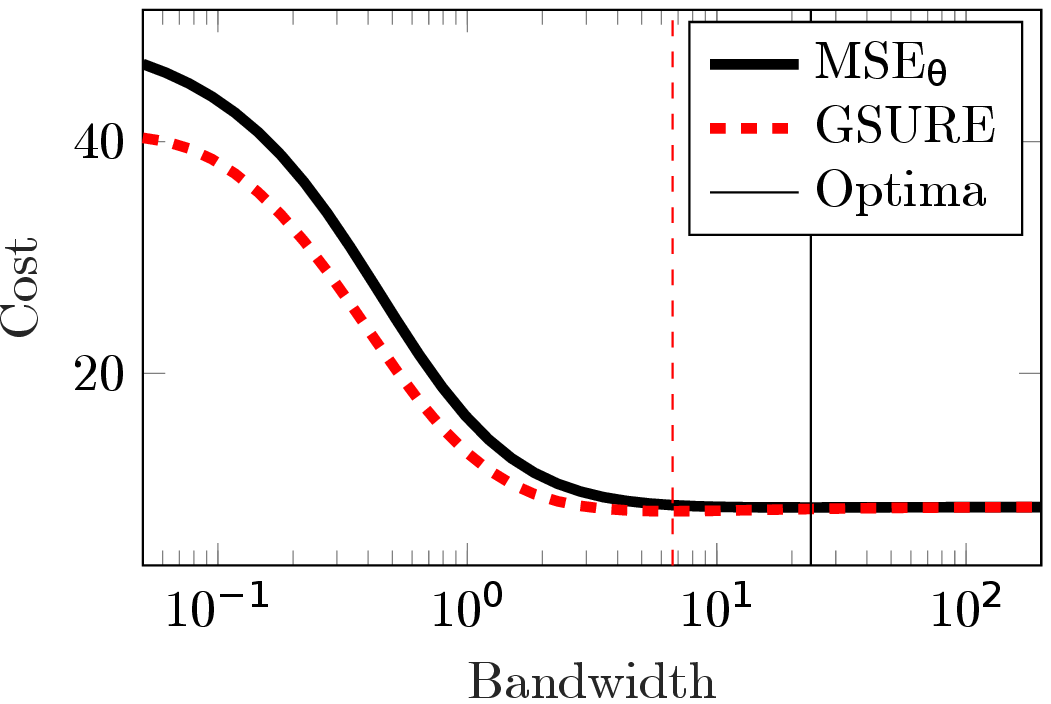}}%
\hfill%
\subfigure[$\tau^\star=0.98$]{\includegraphics[width=0.32\linewidth]{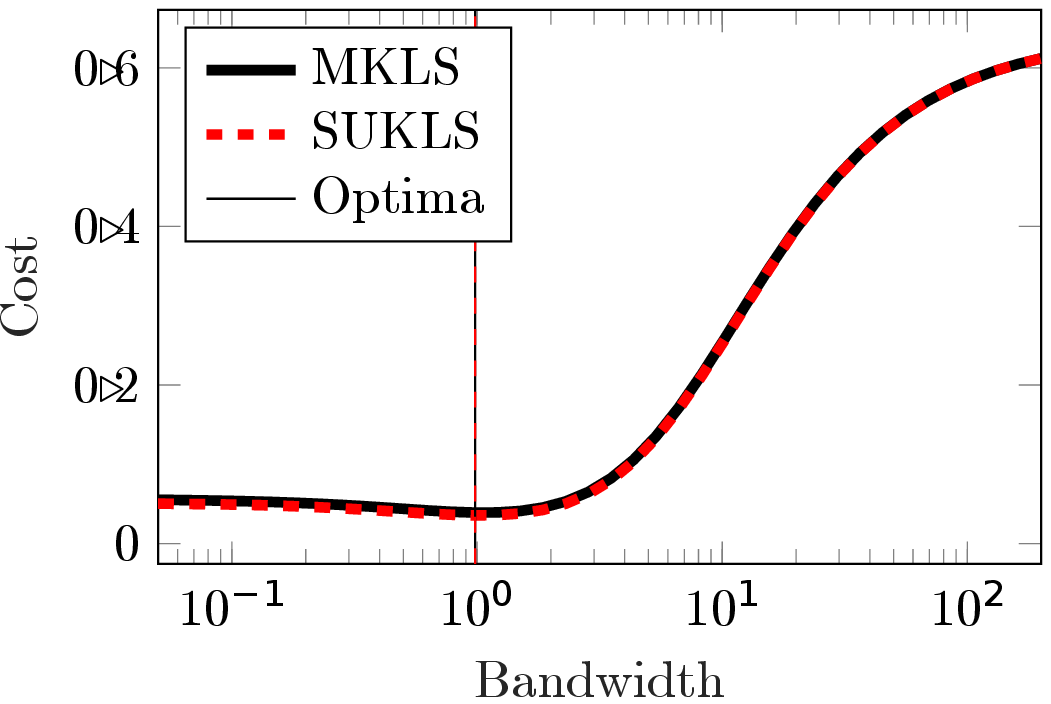}}%
\hfill%
\subfigure[$\tau^\star=1.21$]{\includegraphics[width=0.32\linewidth]{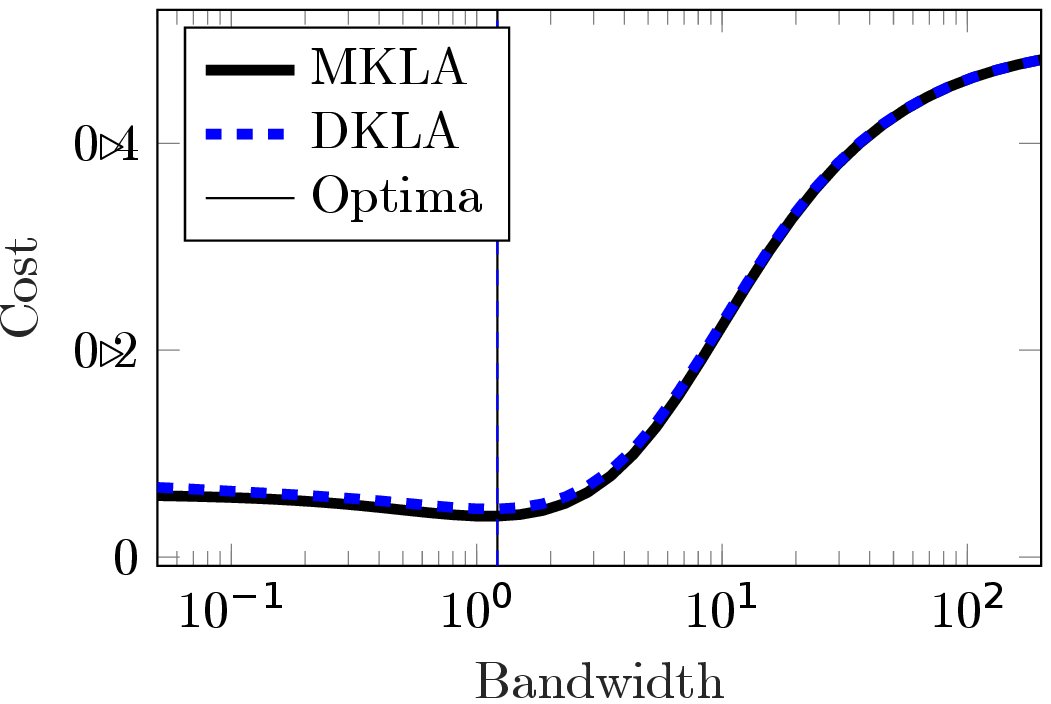}}\\
\caption{
  (a, b) Original and noisy image contaminated by Gamma noise with $L=3$
  (squared root versions of the images are displayed for better visual inspection).
  (c, d, e) Results of non-local filtering
  for the optimal bandwidth with respect to the natural risk $\mathrm{MSE}_\theta$
  $\mathrm{MKLS}$ and $\mathrm{MKLA}$.
  The $\mathrm{MNAE}$ is indicated.
  (f,g,h)
  Risks and their estimates as a function of the bandwidth.
  The optimal bandwidth $\tau^\star$ is indicated.
  Red shows unbiased estimation and blue biased estimation.
}
\label{fig:nlmeans_rayleigh}
\vspace{1em}
\end{figure}

\begin{figure}[!t]
\centering
\subfigure[Original image $\mu$]{\includegraphics[width=0.32\linewidth,viewport=100 70 356 163,clip]{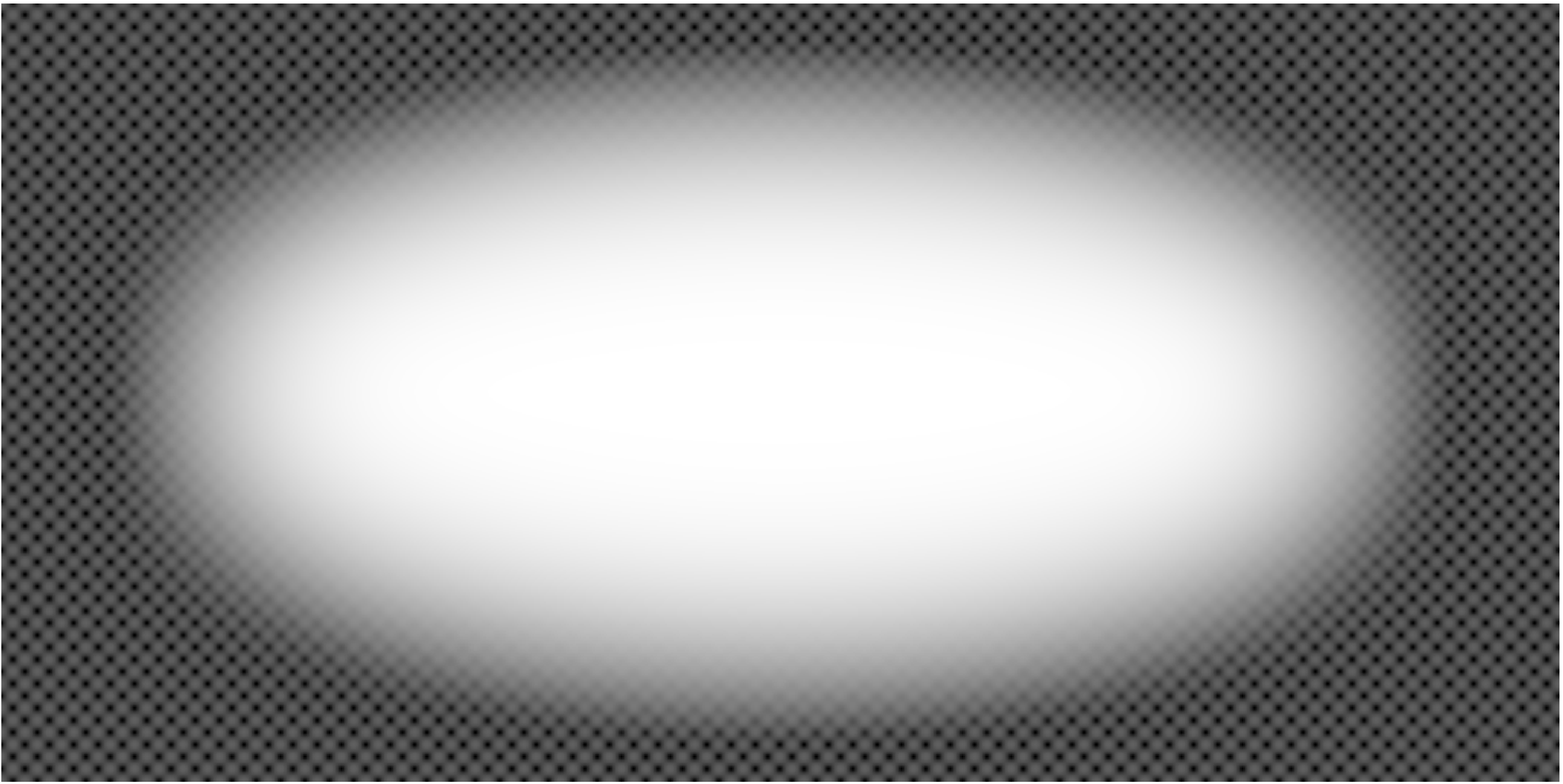}}\hfill%
\subfigure[Noisy image $y$]{\includegraphics[width=0.32\linewidth,viewport=100 70 356 163,clip]{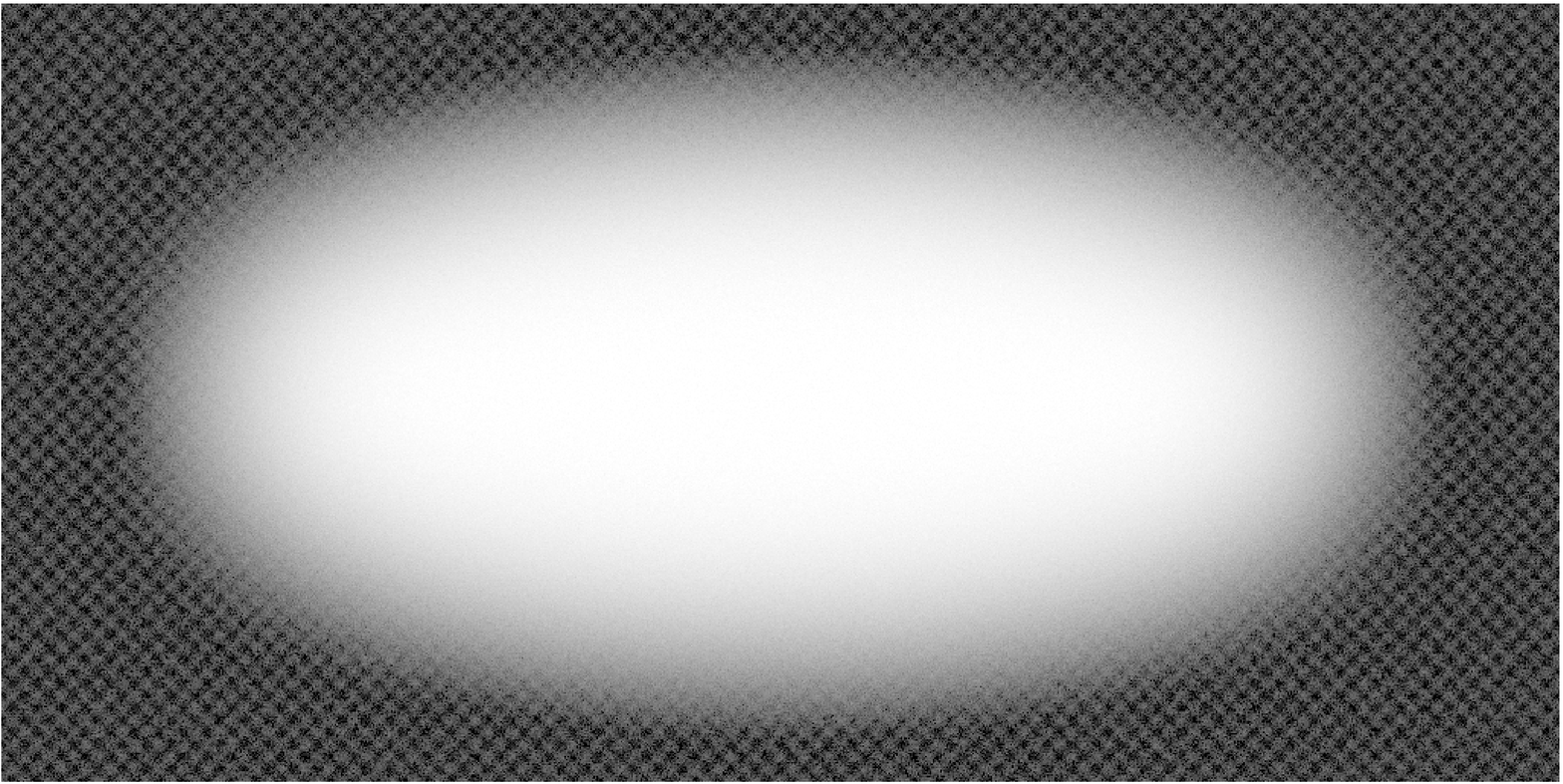}}\hfill%
\hspace{0.32\linewidth}{$ $}\\[-0.5em]
\subfigure[MNAE=0.416]{\includegraphics[width=0.32\linewidth,viewport=100 70 356 163,clip]{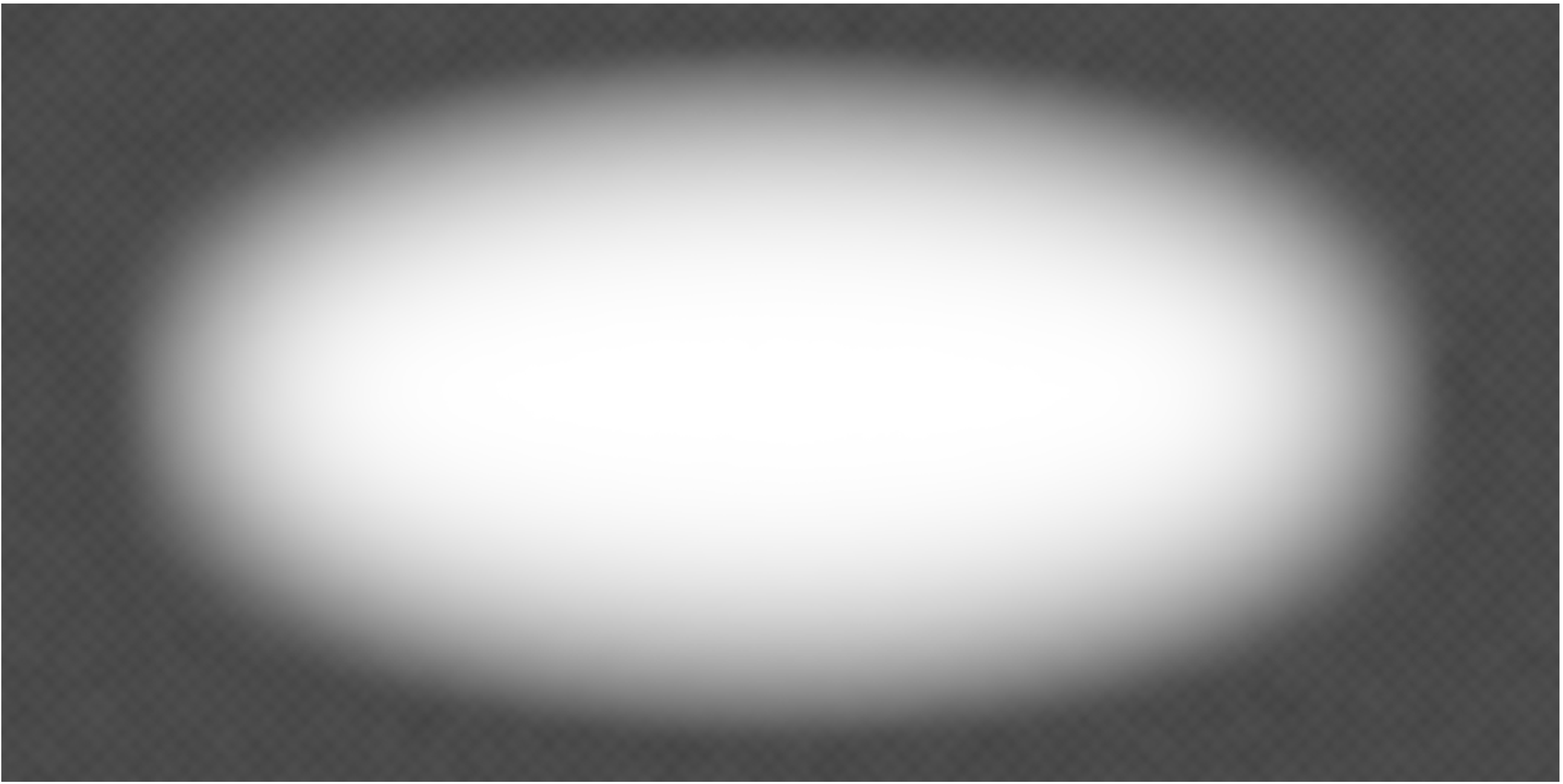}}\hfill%
\subfigure[MNAE=0.261]{\includegraphics[width=0.32\linewidth,viewport=100 70 356 163,clip]{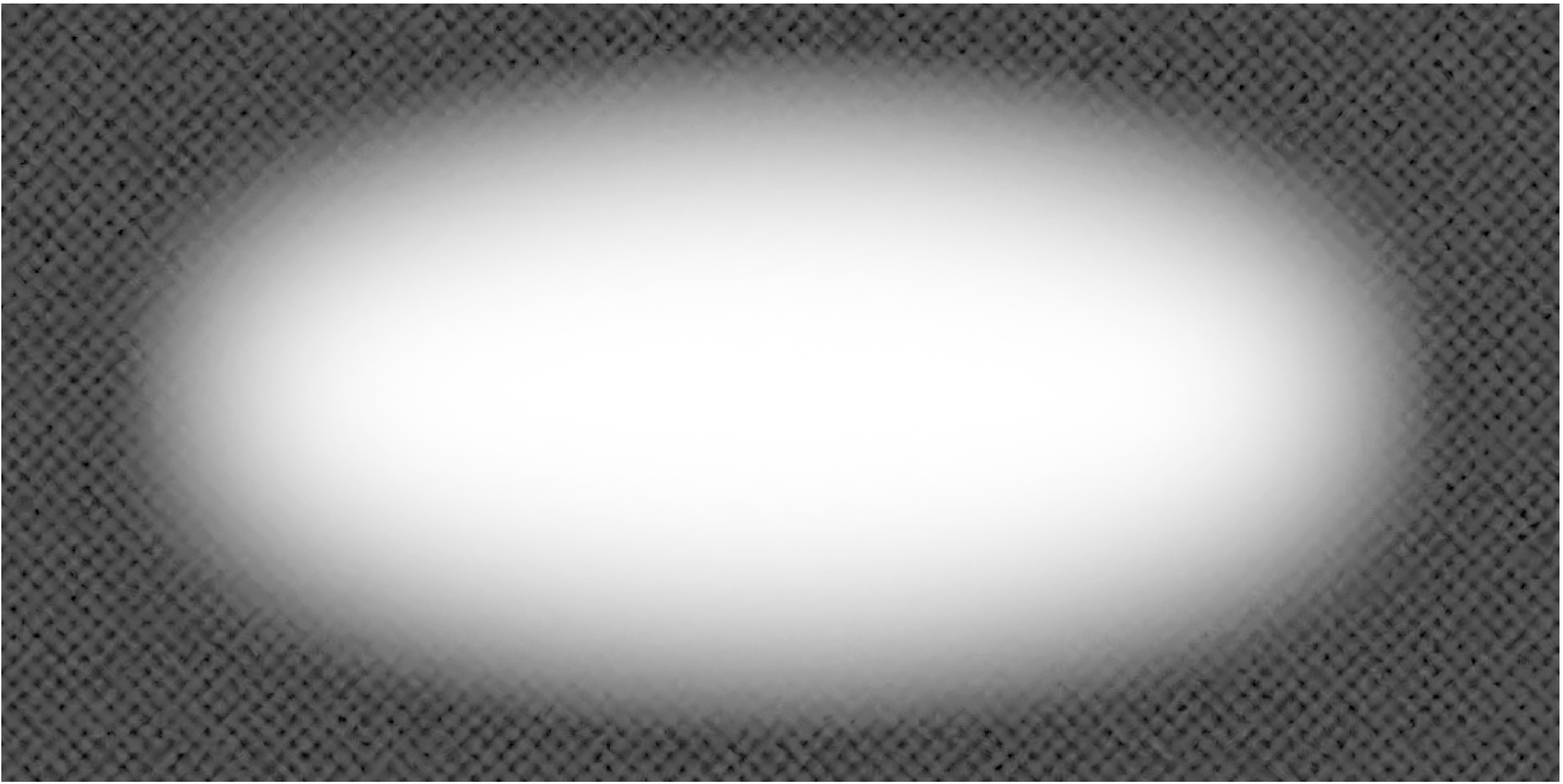}}\hfill%
\subfigure[MNAE=0.261]{\includegraphics[width=0.32\linewidth,viewport=100 70 356 163,clip]{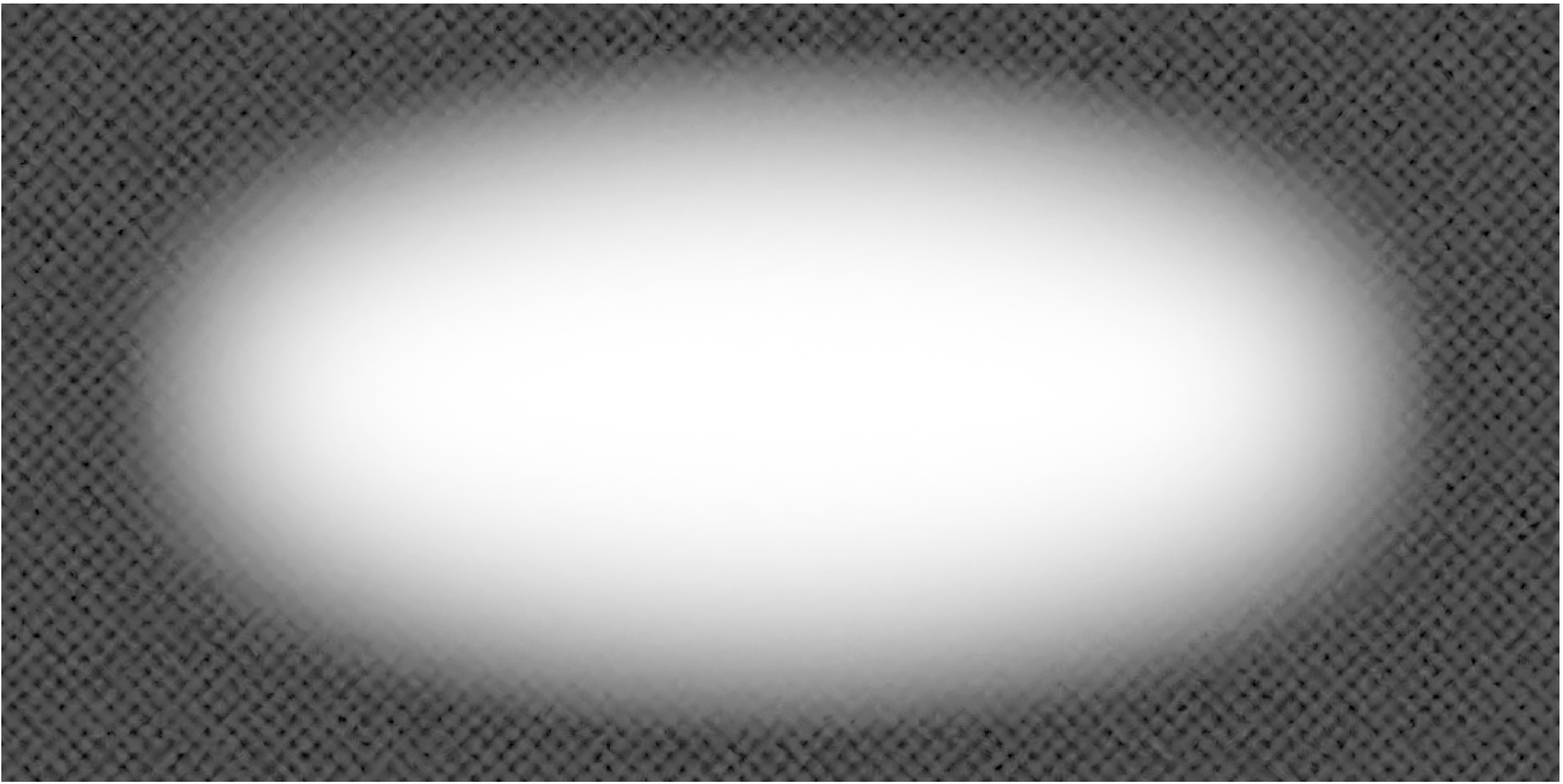}}\\[-0.5em]
\subfigure[$\tau^\star=42.06$]{\includegraphics[width=0.32\linewidth]{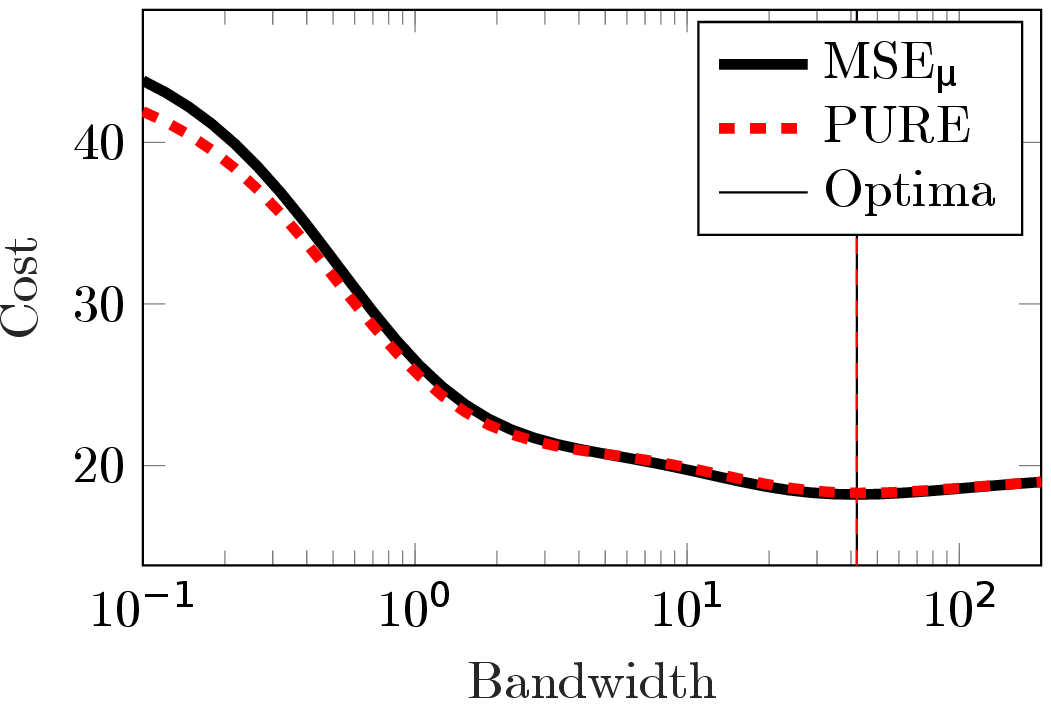}}%
\hfill%
\subfigure[$\tau^\star=0.48$]{\includegraphics[width=0.32\linewidth]{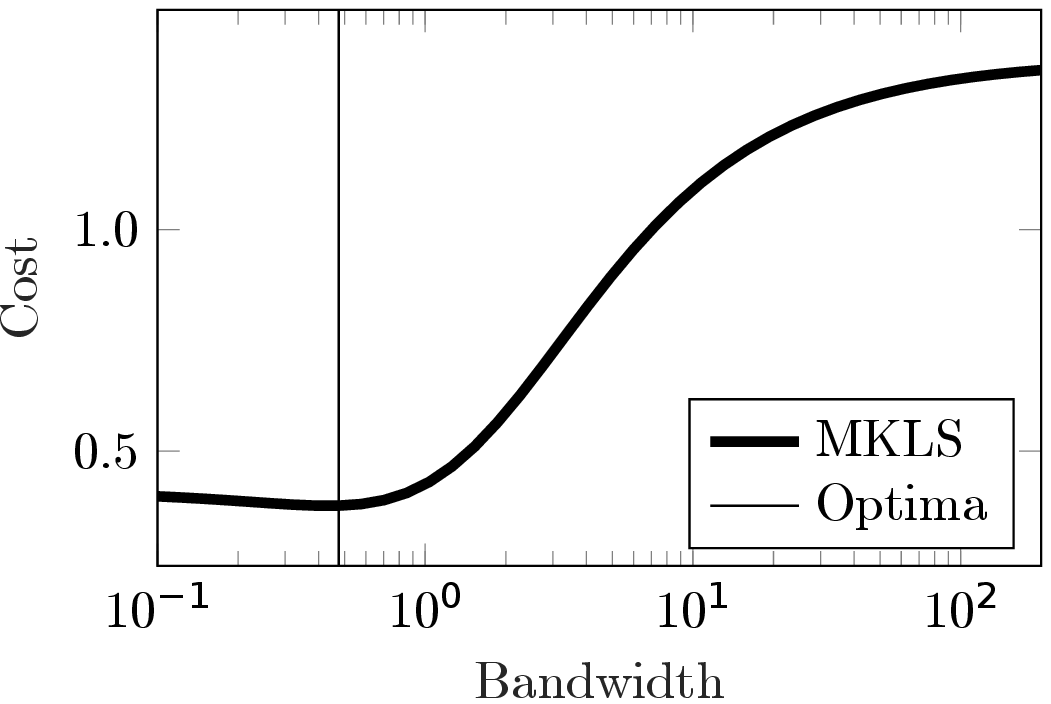}}%
\hfill%
\subfigure[$\tau^\star=0.48$]{\includegraphics[width=0.32\linewidth]{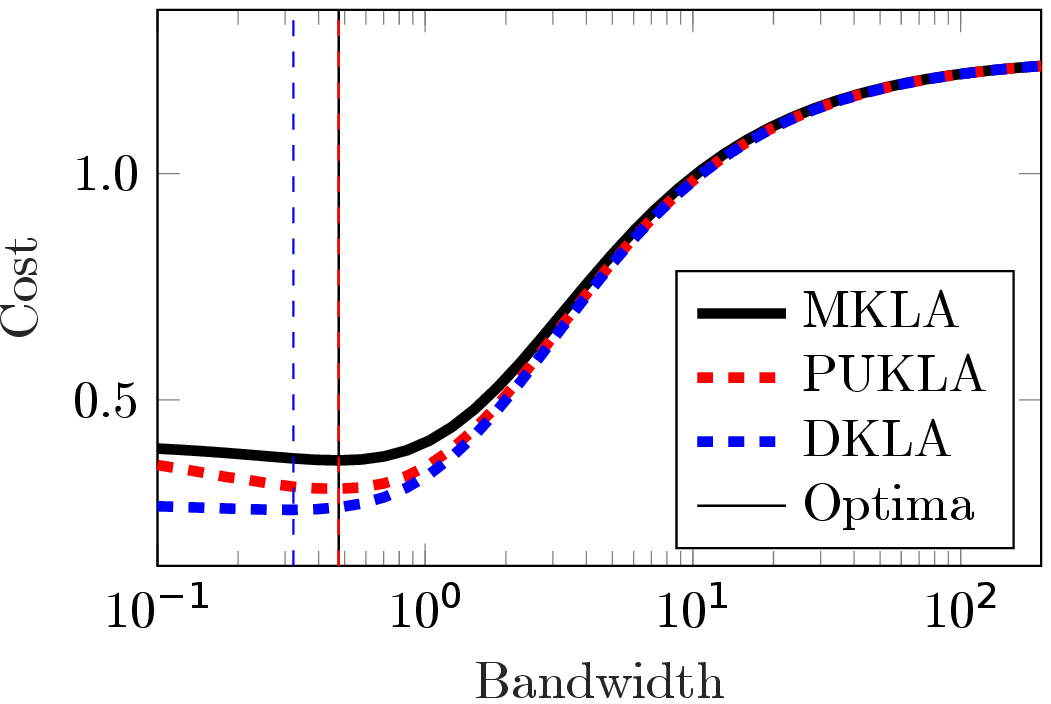}}
\caption{
  (a, b) Original and noisy image contaminated by Poisson noise
  (log of the images are displayed for better visual inspection).
  (c, d, e) Results of non-local filtering
  for the optimal bandwidth with respect to the risk $\mathrm{MSE}_\mu$
  $\mathrm{MKLS}$ and $\mathrm{MKLA}$.
  The $\mathrm{MNAE}$ is indicated.
  (f,g,h)
  Risks and their estimates (when available) as a function of the bandwidth.
  The optimal bandwidth $\tau^\star$ is indicated.
  Red shows unbiased estimation and blue biased estimation.
}
\label{fig:nlmeans_poisson}
\vspace{1em}
\end{figure}

\paragraph{Simulations in non-linear filtering.}

We consider here that $\hat{\mu}$ is the non-local filter \cite{buades2005rid} defined by
\begin{align}
  \hat{\mu}(y) = W(y) y
  \qwithq
  W_{i,j}(y) = \exp\left(-\frac{d(\Pp_i y, \Pp_j y)}{\tau}\right)
\end{align}
where $\Pp_i \in \RR^{p \times n}$ is a linear operator extracting a patch
(a small window of fixed size) at location $\delta_i$,
$d : \RR^p \times \RR^p \to \RR^+$
is a dissimilarity measure (infinitely differentiable and adapted to the exponential family following
\cite{deledalle2012compare})
and $\tau>0$ a bandwidth parameter.
Remark as $W(y) \in \RR^{d \times d}$ depends on $y$, $\hat{\mu}(y)$ is non-linear. In this context,
we will evaluate again the relevance
of the proposed loss functions and their estimates as objectives
to select the bandwidth $\tau$.\\

Figure \ref{fig:nlmeans_rayleigh} gives an example of a noisy observation $y$
of an image $\mu$ representing a bright two dimensional chirp signal shaded gradually into a dark homogeneous region.
The noisy observation $y$ is contaminated by noise following a Gamma
distribution with shape parameter $L=3$.
We have again evaluated the relevance of the natural risk
$\mathrm{MSE}_\theta$ given by $\norm{\mu^{-1} - \hat{\mu}(Y)^{-1}}^2$,
$\mathrm{MKLS}$ and $\mathrm{MKLA}$ in selecting the bandwidth parameter.
Visual inspection of the results obtained at the
optimal bandwidth for each criterion shows that the natural risk fails in selecting
a relevant bandwidth while $\mathrm{MKLS}$ and $\mathrm{MKLA}$ both provide more satisfying results.
As the image $\mu$ is very smooth in the darker region,
the natural risk favors strong variance reduction
leading to a strong smoothing of the texture in the brightest area.
Again, the Kullback-Leibler loss functions find a good trade-off preserving simultaneously
the bright texture and reducing the noise in the dark homogeneous region, as assessed by the MNAE.
Finally, estimators of these loss functions with respectively GSURE, SUKLS and DKLA are given.\\

Figure \ref{fig:nlmeans_poisson} gives a similar example where the image $\mu$ represents
a two dimensional chirp signal
shaded gradually into a bright homogeneous region.
The image is displayed in log-scale to better assess the variations of the texture
in the darkest region.
The noisy observation $y$ is corrupted by independent noise following a Poisson distribution.
We have evaluated the relevance of the risks $\mathrm{MSE}_\mu$,
$\mathrm{MKLS}$ and $\mathrm{MKLA}$ in selecting the bandwidth parameter.
Visual inspection of $y$ shows that darker regions are more affected
by noise than brighter ones. This is due to the fact that Poisson corruptions
lead to a signal to noise ratio evolving as $\sqrt{\mu}$.
It follows that the $\mathrm{MSE}$
essentially penalizes the residual variance of the brightest region hence leading to
a strong smoothing of the texture in the darkest area. Again, Kullback-Leibler losses lead to selecting
a more relevant bandwidth, smoothing less the brightest area but preserving better the texture,
as assessed by the MNAE.
Finally, estimators of the $\mathrm{MSE}$ with PURE and $\mathrm{MKLA}$ with
PUKLA and DKLA are given.
Note that estimators of $\mathrm{MKLS}$ are not available for non-local filtering
under Poisson noise.

\subsection{Application to variable selection}

We now consider the problem of variable selection in linear regression problems, i.e.,
in finding the non-zero components of a vector $\beta \in \RR^q$ that fulfills the
assumption that an observed vector $y \in \RR^d$ has expectation $\mu = X \beta$
where $X \in \RR^{d \times q}$ is the so-called design matrix.
To this aim, we consider the Least Absolute Shrinkage and Selection Operator (LASSO)
\cite{tibshirani1996regre} given, for $\lambda > 0$, by
\begin{align*}
  \hat{\beta}(y) \in \uargmin{\beta \in \RR^q}
  -\log p(y; \theta = \phi(X \beta)) + \lambda \norm{\beta}_1~.
\end{align*}
In this case the predictor $\hat{\mu}$ is given by $\hat{\mu}(y) = X \hat{\beta}(y)$.
The LASSO is known to promote sparse solutions, i.e., such that
the number of non-zero entries of $\hat \beta$ is small compared to $q$.
The level of sparsity is indirectly controlled by the regularization parameter
$\lambda$, the larger $\lambda$ is,
the sparser $\hat{\beta}$ will be.
Finding the optimal parameter $\lambda$, and then selecting the relevant variables
(columns of $X$) explaining $y$, is a challenging problem that can be addressed
by minimizing an estimator of the risk.
In this context,
we will evaluate again the relevance
of the different proposed loss functions and their estimates as objectives
to select a regularization parameter $\lambda$ offering
a relevant selection of variables.

Figure \ref{fig:lasso} and Table \ref{tab:lasso}
provide results obtained on such a linear regression
problem where $X$ is an orthogonal matrix and $d = q = 16,384$. The vector
$\beta$ was chosen such that $28\%$ of its entries are non-zero.
We have generated $200$ independent realizations $y$ of $Y$ using a Gamma
distribution model with scale parameter $L = 8$.
We have again evaluated the relevance of the natural risk
$\mathrm{MSE}_\theta$ given by $\norm{\mu^{-1} - \hat{\mu}(Y)^{-1}}^2$,
$\mathrm{MKLS}$ and $\mathrm{MKLA}$ in selecting the regularization parameter.
Figure \ref{fig:lasso} shows the evolution of these objectives as a function of $\lambda$.
It shows that KL objectives lead to selecting a larger $\lambda$ parameter
than with the natural risk.
Performance in terms of average percentages of
false negatives (FN: $\hat{\beta}_i = 0$ and $\beta_i \ne 0$),
false positives (FP: $\hat{\beta}_i = 0$ and $\beta_i \ne 0$) and
errors (FP or FN) are reported in Table \ref{tab:lasso}.
It shows that tuning the parameter $\lambda$ with respect to KL objectives
leads to lower numbers of errors than with the natural risk.
One can observe that the subsequent LASSO estimators work at different
trade-offs: KL objectives favor FN over FP, while the natural risk
favors FP over FN.
Finally, performances with estimators of $\mathrm{MSE}_\theta$ with GSURE,
$\mathrm{MKLS}$ with SUKLS, and MKLA with DKLA are also given.
It can be observed that risk estimators offer in average comparable results than their
oracle counterparts but have higher variance.
Note that the LASSO is not differentiable, such that DKLA is not guaranteed
to be asymptotically unbiased (as the conditions of Theorem 3 are not fulfilled),
which explains the large discrepancies observed between the results obtained by MKLA and DKLA.
Nevertheless, even though DKLA is not asymptotically unbiased in this case
variable selections with the LASSO guided by DKLA still provides
a good objective for variable selection,
with similar results as if it was guided by the oracle MKLA objective.

\begin{figure}[!t]
  \centering
  \subfigure[]{\includegraphics[width=0.32\linewidth]{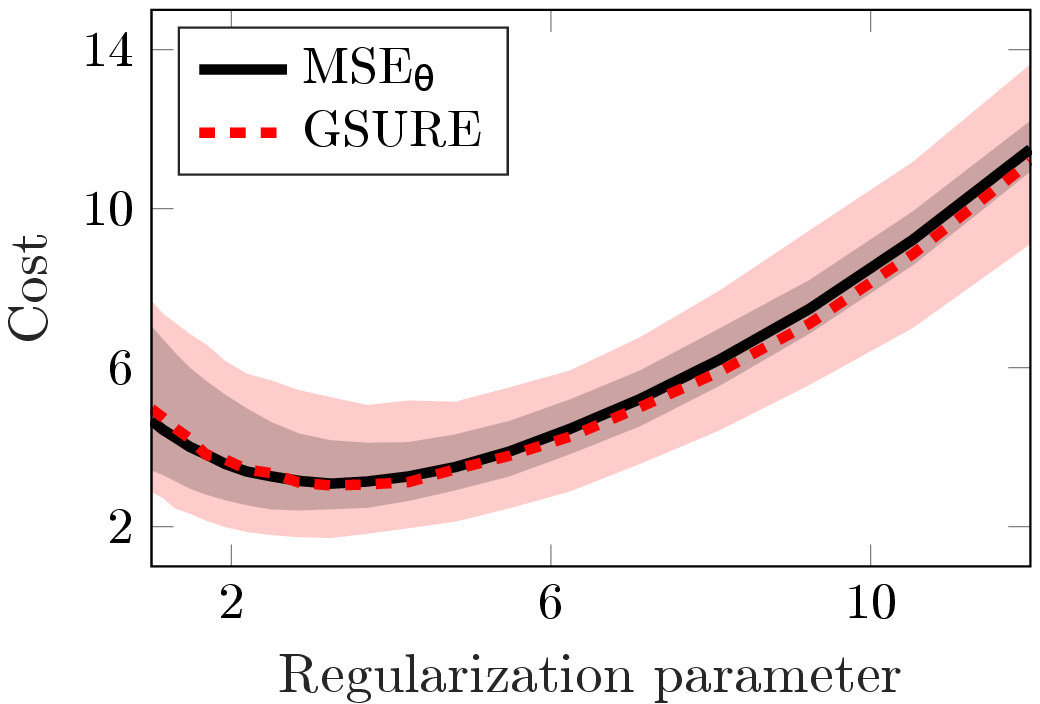}}
  \hfill
  \subfigure[]{\includegraphics[width=0.32\linewidth]{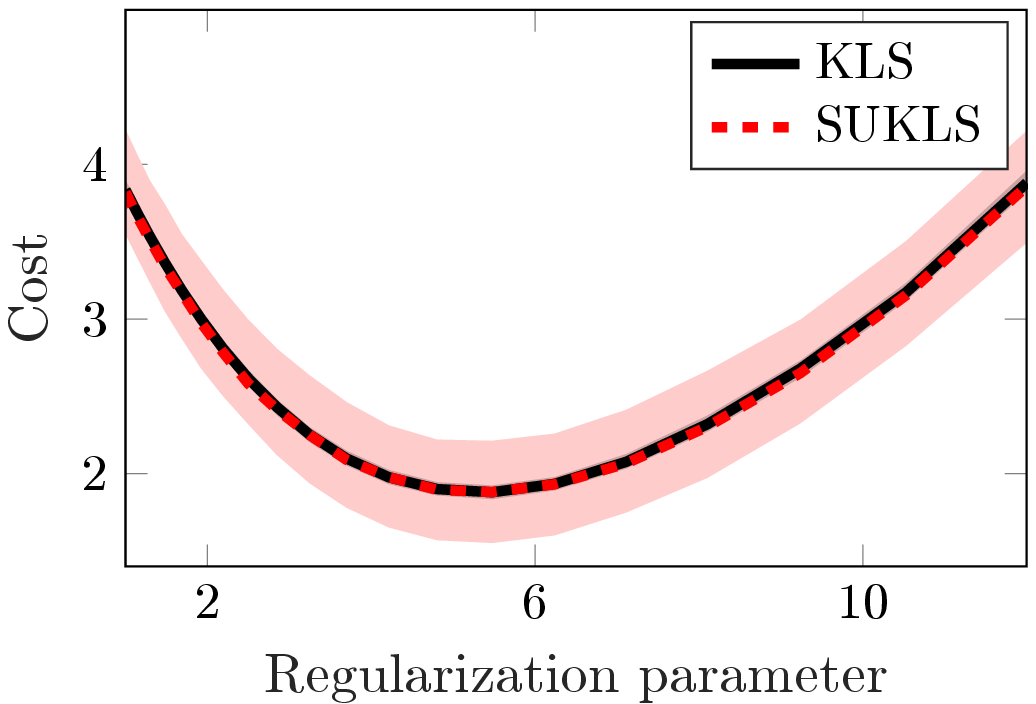}}
  \hfill
  \subfigure[]{\includegraphics[width=0.32\linewidth]{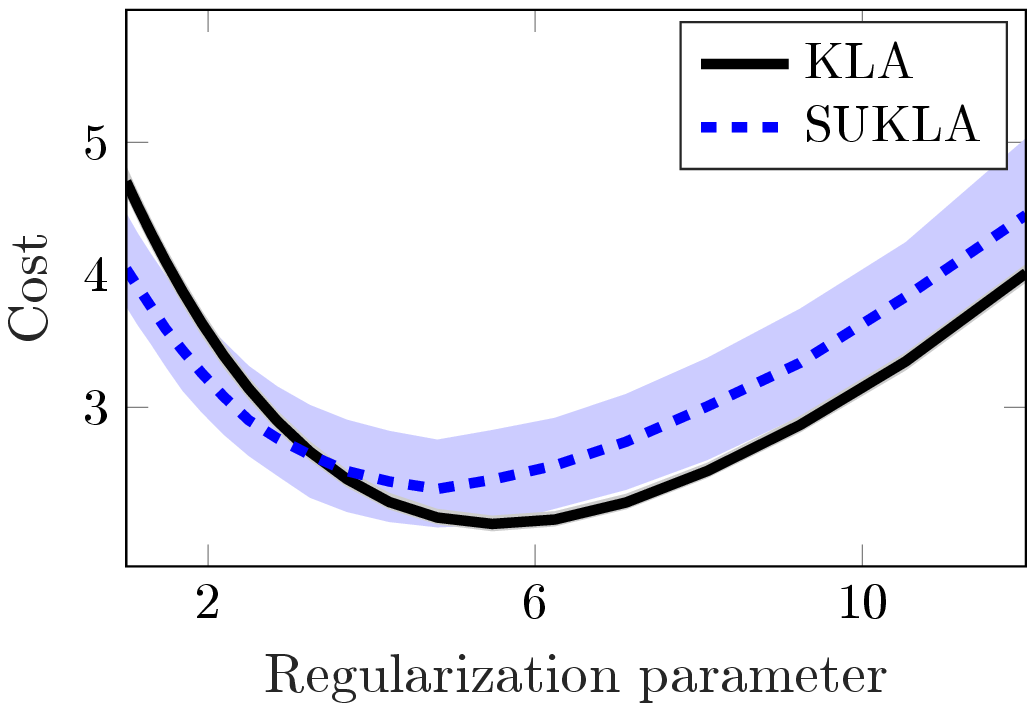}}\\[-1em]
  \caption{Risks and their estimates as a function of the regularization
    parameter $\lambda$ averaged on 200 realizations (their corresponding $90\%$
    confidence intervals are also indicated in shaded colors).
    Red shows unbiased estimation and blue biased estimation.
  }
  \vskip1em
  \label{fig:lasso}
\end{figure}

\begin{table}[!t]
  \centering
  \caption{Performance in terms of errors, false negatives (FN) and false positives (FP)
    of the LASSO in solving a variable selection problem contaminated with
    Gamma noise (with scale parameter $L=8$) and
    guided by different specific objectives. The noise model is considered correctly
    specified ($\hat L = L$) or misspecified
    ($|1 - \hat{L}/L| = 0.1$).
    Results are given in average with plus and minus their standard deviations
    on $200$ noise realizations.
  }
  \label{tab:lasso}
\begin{tabular}{ccccc}
  \hline
  & Errors & FN & FP \\
  \hline
  \hline
  \multicolumn{5}{c}{Correctly specified}\\
  \hline
  MSE$_\theta$ & 34.64 $\pm$ 3.09 & 11.36 $\pm$ 1.28 & 23.27 $\pm$ 4.36 &\\
  GSURE & 34.00 $\pm$ 4.34 & 11.74 $\pm$ 1.95 & 22.25 $\pm$ 6.26 &\\
  \hline
  MKLS & 26.24 $\pm$ 0.28 & 15.68 $\pm$ 0.16 & 10.56 $\pm$ 0.23 &\\
  SUKLS & 26.84 $\pm$ 1.17 & 15.31 $\pm$ 0.82 & 11.52 $\pm$ 1.95 &\\
  \hline
  MKLA & 26.24 $\pm$ 0.28 & 15.68 $\pm$ 0.16 & 10.56 $\pm$ 0.23\\
  DKLA & 28.31 $\pm$ 1.62 & 14.37 $\pm$ 0.95 & 13.94 $\pm$ 2.54\\
  \hline
  \hline
  \multicolumn{5}{c}{Misspecified}\\
  \hline
  GSURE & 35.94 $\pm$ 4.44 & 10.93 $\pm$ 1.81 & 25.01 $\pm$ 6.22\\
  \hline
  SUKLS & 28.67 $\pm$ 1.45 & 14.13 $\pm$ 0.81 & 14.53 $\pm$ 2.22\\
  \hline
  DKLA & 30.20 $\pm$ 1.90 & 13.34 $\pm$ 0.96 & 16.86 $\pm$ 2.84\\
  \hline
\end{tabular}
\vskip2em
\end{table}

A last important question is to know whether our risk estimators are robust
against model misspecification, i.e., when the generative model \eqref{eq:klfamily} is
only approximately known.
Indeed, Lv and Liu \cite{lv2014model}
demonstrated the advantage of using KL divergence principle for model selection problems
in both correctly specified and misspecified models.
Along these lines, we have also shown in Table \ref{tab:lasso}
the results obtained under misspecification. We have chosen to evaluate the performance
of the LASSO guided by the aforementioned estimators of the risk
when the shape parameter $L$ of the Gamma distribution is misestimated
by a factor $0.1\%$: $|1 - \hat{L}/L| = 0.1$. We found that the performance
of all estimators drop in this case.
Nevertheless, their relative performance are preserved:
KL objectives lead to lower numbers of errors than with the natural risk.

\section{Conclusion}

We addressed the problem of using and estimating Kullback-Leibler losses for model
selection in recovery problems involving noise distributed within the exponential family.
Our conclusions are threefold:
1)~Kullback-Leibler losses have shown empirically to be more relevant than squared losses
for model selection in the considered scenarios;
2) Kullback-Leibler losses can be estimated in many cases
unbiasedly or with controlled bias depending on the regularity
of both the predictor and the noise;
3) Even though the estimation is subject to variance and bias, the subsequent selection
has shown empirically to be close to the optimal one associated to the loss being estimated.
Future works should focus on understanding under which conditions such a behavior can be guaranteed.
This includes establishing tighter bounds on the reliability, consistency with
respect to the data dimension $d$ and asymptotic optimality results for some given class of predictors.
Estimation of Kullback-Leibler losses and other discrepancies
(e.g., Battacharyya, Hellinger, Mahanalobis, Rényi or Wasserstein distances/divergences)
beyond the exponential family and requiring less regularity on the predictor
should also be investigated.

\bibliographystyle{imsart-number}
\bibliography{refs}

\end{document}